\documentclass[11pt,letterpaper,onecolumn]{article}
\usepackage[T1]{fontenc}
\usepackage[utf8]{inputenc}
\usepackage{lmodern,microtype}
\usepackage{graphicx}
\usepackage[margin=1in]{geometry}
\usepackage{amsmath,amssymb,amsthm,mathtools}
\usepackage{algorithm,algpseudocode,enumitem,booktabs}
\usepackage{needspace,etoolbox}
\usepackage{tikz}
\usetikzlibrary{arrows.meta}
\usepackage[hidelinks]{hyperref}
\hypersetup{pdftitle={Enumerating small potential maximal cliques with fixed-parameter delay},pdfauthor={Batya Kenig}}
\setlist{itemsep=0.25em,topsep=0.4em}
\newtheorem{theorem}{Theorem}[section]
\newtheorem{lemma}[theorem]{Lemma}
\newtheorem{corollary}[theorem]{Corollary}
\newenvironment{citedLemma}[1]
  {\begin{lemma}[{#1}]}
  {\end{lemma}}
\newcommand{\Local}{\texttt{Local}}
\newcommand{\Regions}{\texttt{Regions}}
\newcommand{\CheckPMC}{\texttt{Check}}
\newcommand{\Compress}{\texttt{Compress}}
\newcommand{\Tor}{\operatorname{Tor}}
\algrenewcommand\textproc{\texttt}
\AtBeginEnvironment{theorem}{\Needspace{8\baselineskip}}
\AtBeginEnvironment{lemma}{\Needspace{8\baselineskip}}
\makeatletter
\providecommand{\theHALG@line}{}
\renewcommand{\theHALG@line}{\thealgorithm.\arabic{ALG@line}}
\makeatother

\newcommand{\nodes}{\texttt{V}}
\newcommand{\edges}{\texttt{E}}
\newcommand{\sminus}{\scalebox{0.75}[1.0]{$-$}}
\def\set#1{\mathord{\{#1\}}}
\newcommand{\eat}[1]{}
\def\PMC{\mathrm{PMC}}

\def\set#1{\mathord{\{#1\}}}
\newcommand{\minlsep}[2]{\mathcal{S}_{#1}(#2)}

\def\eqdef{\mathrel{\stackrel{\textsf{\tiny def}}{=}}}

\newenvironment{repeatresult}[2]
{\vskip0.5em\par\textsc{#1} #2.\em}
{\vskip1em}

\newenvironment{reptheorem}[1]{\begin{repeatresult}{Theorem}{#1}}{\end{repeatresult}}

\title{Enumerating Small Potential Maximal Cliques with Fixed-Parameter Delay}
\author{Batya Kenig}
\date{}

\begin{document}
\maketitle

\begin{abstract}
	A \emph{potential maximal clique (PMC)} of a graph $G$ is a nonempty vertex set
	that forms a maximal clique in some minimal triangulation of $G$.
	Listing the PMCs of a graph is a key step in exact algorithms for
	computing treewidth, minimum fill-in, and related graph parameters.
	In this paper, we present a deterministic algorithm that enumerates
	all PMCs of cardinality at most $k$ in an $n$-vertex graph $G$,
	without repetition, with delay $2^{O(k^3)}n^4\log n$ for $n\ge2$
	and $k\ge1$. This bound also covers the preprocessing time before
	the first output, the time from the last output to termination, and
	the total running time when no such PMC exists. At the core of our
	enumeration algorithm is a procedure that solves the following
	constrained search problem within the same time bound: given sets
	of required and excluded vertices $I,O\subseteq\nodes(G)$, find a
	PMC $\Omega$ of size at most $k$ satisfying $I\subseteq\Omega$ and
	$\Omega\cap O=\emptyset$, or determine that none exists. Our
	enumeration algorithm directly supports these inclusion and exclusion
	constraints as well.
	
	Our approach builds on the \emph{recursive understanding} technique of
	Chitnis et al.~\cite{chitnis2016} and the reduction to \emph{unbreakable
	graphs} of Lokshtanov et al.~\cite{lokshtanov2018}.
	A graph is \emph{$(q,k)$-unbreakable} if no removal of at most $k$ vertices
	leaves two disjoint unions of connected components that each contain
	more than $q$ vertices. We first solve constrained PMC search on
	$(q,k)$-unbreakable graphs, and then apply this procedure to replace
	separated subgraphs with graphs of size $2^{O(k^2)}$ that
	preserve the existence of a constrained solution for every graph
	attached along the shared boundary. By adding a bounded number of
	excluded auxiliary vertices, we ensure that every constrained solution found in the
	reduced graph is a valid PMC of the original graph.
\end{abstract}
\section{Introduction}\label{sec:introduction}

A graph $G$ is \emph{chordal} if every cycle of length four or more
contains a chord, that is, an edge joining two nonconsecutive vertices
of the cycle. A \emph{triangulation} of $G$ is a chordal supergraph
obtained by adding edges to $G$, and it is \emph{minimal} if no proper
subset of the added edges yields a chordal graph~\cite{bouchitte2001}.
A \emph{potential maximal clique} (PMC) of $G$ is a nonempty vertex set
that forms a maximal clique in some minimal triangulation of $G$.
A PMC need not be a clique in $G$ itself: in a four-vertex cycle,
adding either diagonal produces a minimal triangulation consisting of
two triangles, so each of the four three-vertex subsets is a PMC even
though the cycle is triangle-free.
In the \emph{PMC enumeration problem}, the goal is to list all PMCs of
an input graph $G$ without repetition.

Potential maximal cliques serve as the candidate bags for dynamic
programming algorithms that construct optimal tree decompositions and
triangulations. Given the full set of PMCs of $G$, the classical
algorithm of Bouchitt\'e and Todinca~\cite{bouchitte2001} computes both
\emph{treewidth} (one less than the minimum, over all triangulations, of
the maximum clique size) and \emph{minimum fill-in} (the minimum number
of edges needed to triangulate $G$) in time polynomial in $n$ and the
number of PMCs, without explicitly enumerating the minimal
triangulations themselves. Fomin, Todinca, and Villanger~\cite{fomin2015}
substantially generalized this paradigm, showing that a broad class of
induced-subgraph and vertex-deletion problems—including Maximum
Independent Set, Minimum Feedback Vertex Set, and Longest Induced
Path—can likewise be solved in time polynomial in $n$ and the number of
PMCs.

Finding low-cost triangulations and tree decompositions is also central
to numerical linear algebra, probabilistic inference, and database query
processing. In Cholesky factorization of sparse symmetric positive
definite matrices, fill-in edges correspond to new nonzero entries
introduced during Gaussian elimination, so minimizing fill-in directly
reduces storage and can reduce arithmetic cost~\cite{rose1976}. In discrete
Bayesian networks, exact inference via junction trees passes messages
between potentials defined over the maximal cliques of a moralized,
triangulated graph; because a table over $b$ binary variables has $2^b$
entries, the maximum clique size governs both space and time
complexity~\cite{bodlaender2005}. Low-width tree
decompositions similarly define tractable classes of constraint
satisfaction problems and conjunctive queries with projections in
relational databases~\cite{grohe2014}. Motivated by these
applications, Korhonen, Berg, and J\"arvisalo~\cite{korhonen2019jea}
implemented and evaluated PMC-based exact solvers for several graph and
hypergraph width measures as well as total table size in Bayesian
networks; on standard benchmarks, these algorithms were competitive with problem-specific solvers.

Because a graph can have exponentially many PMCs~\cite{brosse2024},
generating the entire family $\PMC(G)$ quickly becomes prohibitive.
Indeed, Korhonen, Berg, and J\"arvisalo~\cite{korhonen2019ijcai}
identified PMC enumeration as the primary scalability bottleneck of the
Bouchitt\'e--Todinca framework in practice. For parameterized and
bounded-width problems, however, full enumeration is unnecessary:
deciding whether the treewidth of $G$ is at most $k-1$, for instance,
requires only the PMCs of size at most
$k$~\cite{korhonen2019ijcai}. In this work, we study the
problem of enumerating all PMCs of cardinality at most $k$ directly,
without generating larger PMCs.

\paragraph{Our results.}
We present a deterministic algorithm that lists all PMCs of size at most
$k$ in an arbitrary graph $G$ with \emph{fixed-parameter delay}, meaning
that the time elapsed before the first output, between any two
consecutive outputs, and after the last output is bounded by
$f(k)\cdot n^c$ for a computable function $f$ and a constant $c$
independent of $k$. We establish this result by reducing enumeration to
the constrained search problem, which we call \texttt{Extend}: given
subsets $I,O\subseteq\nodes(G)$ of required and excluded vertices, find
a PMC $\Omega$ of size at most $k$ such that $I\subseteq\Omega$ and
$\Omega\cap O=\emptyset$, or report that no such PMC exists. 

\begin{theorem}[Main result]\label{thm:intro-main}
Let $G$ be an $n$-vertex graph, where $n\ge2$, let $k\ge1$, and
let $I,O\subseteq \nodes(G)$. There is a deterministic algorithm that
finds a PMC of size at most $k$ containing $I$ and avoiding $O$, or
correctly reports that none exists, in time
\[
 2^{O(k^3)}n^4\log n.
\]
All such PMCs can be enumerated without repetition with delay
$2^{O(k^3)}n^4\log n$. This delay bound also covers the initial and
final waiting times, and thus $M$ outputs are produced in total time
$(M+1)2^{O(k^3)}n^4\log n$.
\end{theorem}

The search bound of Theorem~\ref{thm:intro-main} is established in
Theorem~\ref{thm:general-extend}, and the well-known Lawler--Murty reduction from
constrained search to enumeration~\cite{murty1968,lawler1972} is given at the end of this section.
To our knowledge, Theorem~\ref{thm:intro-main} provides the first
fixed-parameter delay algorithm for enumerating PMCs of size at most
$k$ in general graphs. Previously, Fomin and
Villanger~\cite{fomin2008} gave an algorithm for
listing PMCs of size at most $k$ whose polynomial degree depends on $k$
(yielding an XP running time). Recently, Berg\'e, Limouzy, and
Schivre~\cite{berge2026} presented an
output-linear algorithm for listing \emph{all} PMCs of a connected graph $G$ in total
time $O(n^4\cdot m\cdot |\PMC(G)|)$, where $m=|\edges(G)|$.\eat{, and posed fixed-parameter
output-linear enumeration of size-bounded PMCs as an open problem.}
Theorem~\ref{thm:intro-main} guarantees
both fixed-parameter delay and output-linear total time in the
parameterized setting.

\paragraph{Recognizing a PMC in the whole graph.}
Whether a nonempty vertex subset $\Omega\subseteq\nodes(G)$ is a PMC is a global
property of $G\sminus\Omega$. By the characterization of Bouchitt\'e and
Todinca~\cite{bouchitte2001,brosse2024}, $\Omega$ is a PMC of $G$ if and
only if (1)~no connected component of $G\sminus\Omega$ is adjacent to
every vertex of $\Omega$ (no component is \emph{full}), and (2)~for every
pair of nonadjacent vertices in $\Omega$, some connected component of
$G\sminus\Omega$ contains a neighbor of each vertex.
Consequently, components of $G\sminus\Omega$ play a dual role: they
provide the paths that witness condition~(2), and they can also form
full-component obstructions that violate condition~(1). Naive graph
simplifications, such as deleting vertices or splitting the graph along
a separator, can destroy required connections or inadvertently eliminate
a full component. While testing all vertex subsets of size at most $k$
takes $n^{O(k)}$ time, deciding whether a prescribed vertex set can be
extended to a PMC of unbounded size is
NP-complete~\cite{berge2026}. Our algorithm must therefore
exploit the cardinality bound $k$ to restrict the search space while
faithfully preserving both connectivity witnesses and full-component
obstructions across graph reductions.

\paragraph{Locating a PMC in a small connected region.}
We begin by solving \texttt{Extend} on \emph{$(q,k)$-unbreakable}
graphs—those in which no separator of size at most $k$ partitions the
remaining vertices into two disjoint unions of connected components that
each exceed size $q$. We prove that if $G$ is $(q,k)$-unbreakable and
$n>3q+k$, then every PMC $\Omega$ of size at most $k$ satisfies
$S\subsetneq\Omega\subseteq C\cup S$, where $C$ is a connected vertex
set of size at most $q$ and $S=N_G(C)$ is its external neighborhood of
size at most $k-1$.

To see why this holds, observe that since $|\Omega|\le k$ and $n>3q+k$,
the $(q,k)$-unbreakability of $G$ implies that $G\sminus\Omega$ has a
unique connected component $D$ of size greater than $q$. Set
$S=N_G(D)$. Because $\Omega$ has no full component in $G\sminus\Omega$,
we have $S\subsetneq\Omega$. By the pairwise connection property of
PMCs, any two vertices in $\Omega{\setminus}S$ are either adjacent or
connected via a component of $G\sminus\Omega$ distinct from $D$; in
either case, the connecting path avoids $S$, which places all vertices
of $\Omega{\setminus}S$ inside a single connected component $C$ of
$G\sminus S$. The same pairwise condition ensures that every vertex of
$S$ has a neighbor in $C$, so $N_G(C)=S$. Finally, because $S$ separates
$C$ from the large component $D$ and $|S|\le|\Omega|-1\le k-1$,
unbreakability forces $|C|\le q$. Thus $\Omega$ consists of the entire
neighborhood $S=N_G(C)$ together with a nonempty subset of $C$.
Adapting the \emph{bounded-neighborhood connected-set enumeration} of Fomin
and Villanger~\cite[Section~3]{fomin2008}, we enumerate all candidate
pairs $(C,N_G(C))$ with $|C|\le q$ and $|N_G(C)|\le k-1$ by branching on
whether each newly reached neighbor joins $C$ or $N_G(C)$. Testing the
resulting subsets $\Omega=S\cup X$ (for nonempty $X\subseteq C$ with
$|X|\le k-|S|$) against the global PMC criterion solves \texttt{Extend}
on $(q,k)$-unbreakable graphs in time
$2^{O(k\log(2+q/k))}(n+1)^3$ (Theorem~\ref{thm:main}). Crucially, the
dependence on $q$ in the exponent is logarithmic: in our general
reduction, $q$ will be $2^{O(k^2)}$, so the factor
$2^{O(k\log(2+q/k))}$ becomes $2^{O(k^3)}$.

\paragraph{Replacing a separated piece.}
To handle arbitrary graphs, we apply the \emph{recursive understanding}
framework~\cite{chitnis2016,lokshtanov2018}. When $G$ admits a small
separation, we decompose $G$ into two induced subgraphs $H$ and $J$
that overlap only on a small boundary $Z=\nodes(H)\cap\nodes(J)$, and
we replace the piece $H$ with a graph $K$ whose size is bounded by a
function of $k$. For every outside graph $J$ attached along $Z$, the graph obtained by
attaching $J$ to $H$ (via $Z$) must have a PMC of size at most $k$ satisfying its
inclusion and exclusion constraints if and only if the graph obtained
by attaching the same $J$ to $K$ has such a PMC under the corresponding
constraints. Thus, replacing $H$ by $K$ preserves whether a constrained
solution exists, regardless of the graph attached outside the boundary.
We construct $K$ in two stages: first, given a subset $R\subseteq\nodes(H)$
of retained vertices (containing $Z$ and all required vertices of $H$),
we compress $H$ into a graph $K$ of size at most $2|R|$ that preserves
PMC membership for all candidates $\Omega$ satisfying
$\Omega\cap\nodes(H)\subseteq R$; second, we compute a set $R$ of size
$2^{O(k^2)}$ such that for every outside graph $J$ and every choice of
constraints on $J$, if $H\cup J$ admits a constrained PMC of size at
most $k$, then it admits one whose intersection with $\nodes(H)$ is
contained in $R$.

\paragraph{Preserving PMCs under torso compression.}
Given the retained set $R\subseteq\nodes(H)$, our starting point is the
\emph{torso} $\Tor(H,R)$, obtained from $H[R]$ by turning the
neighborhood $N_H(C)\subseteq R$ of each connected component $C$ of
$H\sminus R$ into a clique~\cite{marx2013}. While the torso preserves
all paths through $H\sminus R$ whose endpoints lie in $R$, it does not
preserve PMC membership on its own: if a deleted component $C$ of
$H\sminus R$ had $N_H(C)=\Omega$, then $C$ was a full component for
$\Omega$ in $H$, and deleting $C$ removes this obstruction. For example, in the
three-vertex path $a$--$c$--$b$ with $R=\set{a,b}$, the
torso is the single edge $(a,b)$, for which $\set{a,b}$ is a PMC even
though the deleted component $\set{c}$ makes $\set{a,b}$ invalid in $H$.
Lemma~\ref{lem:torso} shows that, for any outside graph $J$ attached to $H$
via $Z$, a candidate $\Omega\subseteq R\cup\nodes(J)$ is a PMC
before replacing $H$ by $\Tor(H,R)$ if and only if it is a PMC
after making the replacement and $\Omega\ne N_H(C)$ for every component $C$ of
$H\sminus R$.

The neighborhoods requiring this repair are both cliques and PMCs
of the torso $\Tor(H,R)$. By a bound of Fomin and
Villanger~\cite{fomin2013}, there are at most $|R|$
such neighborhoods. We therefore add at most $|R|$ auxiliary vertices,
giving a replacement graph with at most $2|R|$ vertices. For every
outside graph attached along the boundary, this replacement preserves
whether each candidate formed from vertices of $R$ and of the outside
graph is a PMC (Theorem~\ref{thm:compression}).
\paragraph{Representing attachments by finitely many tests.}
To select the retained set $R\subseteq\nodes(H)$ when $H$ is
$(q,k)$-unbreakable, we must ensure that $R$ captures at least one
witness $\Omega\cap\nodes(H)$ for every outside graph $J$ and every
feasible intersection $X=\Omega\cap\nodes(J)$ of size at most $k$.
In Lemma~\ref{lem:retained-small}, we show that once $X\subseteq\nodes(J)$
is fixed, the effect of $J$ on candidates $\Omega$ with
$\Omega\cap\nodes(J)=X$ can be simulated by a small graph $J_X$ on
$Z\cup X$ together with at most one excluded auxiliary vertex: specifically,
$J_X$ is formed from the torso $\Tor(J,Z\cup X)$ by adding a single
excluded vertex adjacent to $X$ whenever $J\sminus(Z\cup X)$ contains a
component whose neighborhood in $J$ is exactly $X$. Thus any outside
attachment $J$ and selected set $X$ of size at most $k$ can be
represented by a \emph{test graph} $T$ consisting of the boundary $Z$ and at
most $k+1$ additional vertices.

We choose $R$ using a small collection of test-graphs that capture how a
solution can interact with the rest of the graph. For each test, we apply our algorithm for constrained PMC search on
unbreakable graphs to find one solution, if one exists, and retain
the vertices of $H$ used by that solution. We also keep the
required vertices and the vertices through which $H$ connects to the
rest of the graph. Our construction yields a set $R$ of size
$2^{O(k^2)}$ with the following guarantee: whenever attaching an
outside graph to $H$ admits a solution, at least one such solution
uses only vertices of $R$ within $H$ (Theorem~\ref{thm:retained}).

\paragraph{Reducing a general graph.}
The construction described above gives a small replacement for an unbreakable
graph. To handle a general graph, we use a procedure of Chitnis et
al.~\cite{chitnis2016}, in the form given by Golovach et
al.~\cite{golovach2018}, that either certifies the required
unbreakability condition or divides the graph into two smaller pieces
that meet in at most $k$ vertices. In the latter case, we first reduce
one piece, applying the same process within it if necessary. Once
we obtain a small replacement for that piece, we substitute the piece with this small replacement
and continue reducing the resulting graph.

The pieces are chosen large enough that each substitution decreases
the number of vertices. We prove that only $O(n)$ applications of
this reduction are needed, and that the entire input can be replaced
by a graph with $2^{O(k^2)}$ vertices. We then search this small graph
directly for existence of the constrained PMC. Each replacement preserves the existence of a constrained
solution. Moreover, auxiliary vertices cannot be selected, and our
preservation guarantees ensure that every returned solution is a valid
constrained PMC of the original graph.

\paragraph{From constrained search to enumeration.}
With \texttt{Extend} in hand, we enumerate all constrained PMCs of size
at most $k$ via the standard Lawler--Murty partitioning scheme~\cite{murty1968,lawler1972}.
To avoid an extra factor of $n$ in the branching factor (and thus in the
delay), we rely on the basic property that no PMC properly contains another. Indeed, if $A\subsetneq B$ were both
PMCs of $G$, any vertex $v\in B{\setminus}A$ would lie in some connected
component $D$ of $G\sminus A$. Because $A$ is a PMC, $D$ cannot be full
for $A$, so some vertex $u\in A$ has no neighbor in $D$. Every
$v,u$-path in $G$ must therefore intersect $A$ before reaching $u$, which
means every $v,u$-path contains an internal vertex in $A\subsetneq B$, which contradicts the pairwise connection property for $B$ (Lemma~\ref{thm:pmc}).

\Needspace{6\baselineskip}
Given a solution $Q\in\PMC_k(G,I,O)$ returned by \texttt{Extend}, let
$Q{\setminus}I=\{x_1,\ldots,x_t\}$, where $t\le k$. Since no other PMC in
$\PMC_k(G,I,O)$ can contain $Q$, every remaining solution must exclude
at least one vertex of $Q{\setminus}I$. Grouping remaining solutions by
the first index $j\in\{1,\ldots,t\}$ such that $x_j$ is excluded
partitions $\PMC_k(G,I,O){\setminus}\{Q\}$ into $t$ disjoint subproblems
defined by the constraint pairs
\[
 I_j=I\cup\{x_1,\ldots,x_{j-1}\},\qquad
 O_j=O\cup\{x_j\}.
\]
We maintain a stack of active subproblems, each paired with a witness
solution already computed by \texttt{Extend}. After one initial call to
\texttt{Extend}$(G,k,I,O)$, each step pops a pair $((I',O'),Q)$ from
the stack, outputs $Q$, calls \texttt{Extend}$(G,k,I'_j,O'_j)$ for each
$j\in\{1,\ldots,|Q{\setminus}I'|\}$, and pushes those subproblems that
return a witness. Because at most $k$ calls to \texttt{Extend} are made
between consecutive outputs and each child increments $|I'\cup O'|$, the
enumeration outputs every member of $\PMC_k(G,I,O)$ without repetition
and with delay $k\cdot 2^{O(k^3)}n^4\log n = 2^{O(k^3)}n^4\log n$. This establishes Theorem~\ref{thm:intro-main}.

\paragraph{Organization.}
The rest of this paper is organized as follows. After introducing
notation and preliminary definitions in Section~\ref{sec:preliminaries},
we present and analyze the constrained search algorithm $\Local$ for
$(q,k)$-unbreakable graphs in Sections~\ref{sec:local-algorithm}--\ref{sec:proof},
and briefly discuss the role of the unbreakability threshold $q$ in
Section~\ref{sec:q}. Section~\ref{sec:general} outlines our two-step
piece-replacement strategy for breakable graphs. We then establish the
torso-compression step for a fixed retained set in
Section~\ref{sec:compression}, show how to compute the retained set via
bounded test graphs in Section~\ref{sec:retained}, and combine these
components into the general recursive reduction in
Section~\ref{sec:general-algorithm}.
\section{Preliminaries}\label{sec:preliminaries}

\paragraph{Graphs, paths, and neighborhoods.}
All graphs considered in this paper are finite, simple, and undirected.
For a graph $G$, we denote its vertex and edge sets by $\nodes(G)$ and
$\edges(G)$, respectively, and write $n=|\nodes(G)|$ when $G$ is clear
from context. We write either $uv$ or $(u,v)$ for an undirected edge
with endpoints $u$ and $v$. For a vertex subset $A\subseteq\nodes(G)$,
$G[A]$ denotes the subgraph of $G$ induced by $A$, with
$\nodes(G[A])=A$ and $\edges(G[A])=\{uv\in\edges(G):u,v\in A\}$, and
$G\sminus A=G[\nodes(G){\setminus}A]$ denotes the subgraph obtained by
deleting $A$.

A \emph{path} is a sequence of distinct vertices in which consecutive
vertices are adjacent; an \emph{$s,t$-path} is a path with endpoints $s$
and $t$, and its \emph{internal vertices} are all vertices on the path
other than $s$ and $t$. A \emph{walk} is a sequence of adjacent vertices
in which vertices may repeat; shortcutting repeated vertices in any
$s,t$-walk yields an $s,t$-path whose vertex set is contained in that of
the walk. A \emph{cycle} is a closed sequence of at least three distinct
vertices in which consecutive vertices (including the last and first)
are adjacent. A \emph{chord} of a cycle is an edge joining two
nonconsecutive vertices of the cycle, and an \emph{induced} (or
\emph{chordless}) cycle is a cycle with no chords.

For a vertex $v\in\nodes(G)$, let
$N_G(v)=\{u\in\nodes(G):uv\in\edges(G)\}$ denote its open neighborhood.
For a subset $A\subseteq\nodes(G)$, its \emph{external neighborhood} is
\[
 N_G(A)=\left(\bigcup_{a\in A}N_G(a)\right){\setminus}A.
\]
A graph is \emph{connected} if it contains a path between every pair of
its vertices, and a nonempty subset $C\subseteq\nodes(G)$ is
\emph{connected} if $G[C]$ is connected. A \emph{connected component}
(or simply \emph{component}) of $G$ is a maximal connected subset of
$\nodes(G)$. For $A\subseteq\nodes(G)$ and $v\in\nodes(G){\setminus}A$,
we denote by $C_v(G\sminus A)$ the connected component of $G\sminus A$
that contains $v$. A component $C$ of $G\sminus A$ is called a
\emph{full component} of $A$ (or \emph{full with respect to $A$}) if
$N_G(C)=A$.

\paragraph{Potential maximal cliques.}
A \emph{clique} is a set of pairwise adjacent vertices, and it is
\emph{maximal} if it is not a proper subset of another clique. A graph
is \emph{chordal} if it contains no induced cycle of length four or
more. A \emph{triangulation} of $G$ is a chordal graph on $\nodes(G)$
whose edge set contains $\edges(G)$, and it is \emph{minimal} if no
proper subset of the added edges produces a chordal graph.
A nonempty subset $\Omega\subseteq\nodes(G)$ is a \emph{potential
maximal clique} (PMC) of $G$ if $\Omega$ is a maximal clique in at least
one minimal triangulation of $G$. We denote the set of all PMCs of $G$
by $\PMC(G)$. Our algorithms rely on the following characterization due
to Bouchitt\'e and Todinca.
\begin{citedLemma}{\cite[Theorem~3.15]{bouchitte2001},\cite[Theorem~2.1]{brosse2024}}\label{thm:pmc}
	A nonempty set $\Omega\subseteq \nodes(G)$ is a PMC of $G$ if and only if:
	\begin{enumerate}
		\item no component of $G\sminus \Omega$ is full with respect to $\Omega$;
		\item for every two nonadjacent vertices $u,v\in\Omega$, some
		connected component $D$ of $G\sminus \Omega$ satisfies $u,v\in N_G(D)$.
	\end{enumerate}
\end{citedLemma}
We refer to condition~2 of Lemma~\ref{thm:pmc} as the \emph{pair
condition}. In particular, the pair condition implies that any two
distinct vertices $u,v\in\Omega$ are connected in $G$ by a $u,v$-path
whose internal vertices lie entirely outside $\Omega$: if $uv\in\edges(G)$,
the edge $uv$ itself is such a path; otherwise, $u$ and $v$ can be
connected through the component $D$ of $G\sminus\Omega$ supplied by the
pair condition.

\paragraph{Inclusion and exclusion constraints.}
Given sets of required vertices $I\subseteq\nodes(G)$ and excluded
vertices $O\subseteq\nodes(G)$, we define
\[
 \PMC(G,I,O)=\{\Omega\in\PMC(G):
 I\subseteq\Omega,\ \Omega\cap O=\emptyset\},
\]
and for an integer $k\ge0$,
\[
 \PMC_k(G,I,O)=\{\Omega\in\PMC(G,I,O):|\Omega|\le k\}.
\]
The constrained search problem $\texttt{Extend}(G,k,I,O)$ takes
$(G,k,I,O)$ as input and asks to return a member
$\Omega\in\PMC_k(G,I,O)$ (which we call a \emph{solution} or
\emph{witness}), or output $\texttt{No}$ if $\PMC_k(G,I,O)=\emptyset$.
Note that $\PMC_k(G,I,O)=\emptyset$ whenever $I\cap O\ne\emptyset$,
$|I|>k$, $k=0$, or $\nodes(G)=\emptyset$.

\paragraph{Minimal separators.}
For distinct vertices $s,t\in\nodes(G)$, a subset
$X\subseteq\nodes(G){\setminus}\{s,t\}$ is an \emph{$s,t$-separator} if
$s$ and $t$ belong to distinct connected components of $G\sminus X$, and
it is a \emph{minimal $s,t$-separator} if no proper subset of $X$
separates $s$ from $t$. A subset $X\subseteq\nodes(G)$ is a
\emph{minimal separator} of $G$ if it is a minimal $s,t$-separator for
some pair $s,t\in\nodes(G)$. We denote by $\minlsep{}{G}$ the set of all
minimal separators of $G$, and by $\minlsep{k}{G}$ those of cardinality
at most $k$.

\begin{lemma}\label{lem:fullComponents}
 An $s,t$-separator $X$ is minimal if and only if
 $C_s(G\sminus X)$ and $C_t(G\sminus X)$ are full components
 with respect to $X$; that is,
 \[
 N_G(C_s(G\sminus X))=X=N_G(C_t(G\sminus X)).
 \]
\end{lemma}
\begin{proof}
If $X$ is a minimal $s,t$-separator, every $x\in X$ must have a neighbor
in both $C_s(G\sminus X)$ and $C_t(G\sminus X)$; otherwise,
$X{\setminus}\{x\}$ would still separate $s$ from $t$. Conversely, if
both $C_s(G\sminus X)$ and $C_t(G\sminus X)$ are full with respect to
$X$, then putting any $x\in X$ back into $G\sminus X$ connects the two
components via $x$, so no proper subset of $X$ is an $s,t$-separator.
\end{proof}
By Lemma~\ref{lem:fullComponents}, a vertex set $X$ is a minimal
separator of $G$ if and only if $G\sminus X$ contains at least two
distinct full components with respect to $X$.

\paragraph{Separations and unbreakability.}
A \emph{separation} of $G$ is a triple $(L,Z,U)$ of pairwise disjoint
vertex subsets such that $L\cup Z\cup U=\nodes(G)$ and no edge of $G$
has one endpoint in $L$ and the other in $U$. The sets $L$ and $U$ are
the \emph{sides} of the separation (neither of which is required to be
connected), $Z$ is the \emph{separator}, and $|Z|$ is the \emph{order}
of the separation. For integers $q,k\ge0$, a graph $G$ is
\emph{$(q,k)$-unbreakable} if every separation $(L,Z,U)$ of $G$ of order
$|Z|\le k$ satisfies $|L|\le q$ or $|U|\le q$; otherwise, $G$ is
\emph{$(q,k)$-breakable}.

\paragraph{Attaching graphs.}
For two graphs $H$ and $J$ whose vertex sets intersect in
$Z=\nodes(H)\cap\nodes(J)$, we denote by $H\oplus_Z J$ the union graph
with
\[
\nodes(H\oplus_Z J)=\nodes(H)\cup\nodes(J),
\qquad
\edges(H\oplus_Z J)=\edges(H)\cup\edges(J).
\]
We say that $J$ is \emph{attached to $H$ along the boundary $Z$}, and
refer to $H$ as a \emph{piece} of $H\oplus_Z J$. Note that $H[Z]$ and
$J[Z]$ need not induce the same edges on $Z$; the attachment
$H\oplus_Z J$ inherits every edge on $Z$ present in either graph. In
particular, any separation $(L,Z,U)$ of $G$ decomposes $G$ as
\[
G=G[L\cup Z]\oplus_Z G[U\cup Z].
\]
When $H$ and $J$ carry inclusion and exclusion constraint pairs
$(I_H,O_H)$ and $(I_J,O_J)$, respectively, the combined instance on
$H\oplus_Z J$ uses constraints $(I_H\cup I_J,O_H\cup O_J)$.

Whenever we replace a piece $H$ with a compressed graph $K$, retained
vertices of $H$ keep their original identifiers, whereas newly added
\emph{auxiliary vertices} receive fresh identifiers disjoint from all
previously used vertex names. In our correctness proofs for arbitrary
outside attachments $J$, we may assume without loss of generality that
the vertices of $\nodes(J){\setminus}Z$ are renamed away from any
auxiliary vertices introduced during intermediate recursive steps:
because $J$ meets $H$ and $K$ solely at $Z$, such a bijection on
$\nodes(J){\setminus}Z$ preserves all edges, constraints, and PMCs.

\paragraph{Torsos.}
Given a graph $H$ and a subset $R\subseteq\nodes(H)$, the \emph{torso}
of $H$ on $R$, denoted $\Tor(H,R)$, is the graph on vertex set $R$
obtained from $H[R]$ by turning $N_H(C)$ into a clique for every
connected component $C$ of
$H\sminus R$~\cite[Section~2.3, Definition~2.5]{marx2013}. Equivalently,
two distinct vertices $u,v\in R$ are adjacent in $\Tor(H,R)$ if and only
if $H$ contains a $u,v$-path whose internal vertices all belong to
$\nodes(H){\setminus}R$.

\paragraph{Algorithmic conventions.}
A problem parameterized by $p$ is \emph{fixed-parameter tractable} (FPT)
if it can be solved in time $f(p)(n+1)^c$ for a computable function $f$
and a constant $c$ independent of $p$. An enumeration algorithm lists a
target family without repetition and operates with \emph{fixed-parameter
delay} $f(p)(n+1)^c$ if the time before the first output, between any
two consecutive outputs, and after the final output (or the entire
running time when the output family is empty) is at most $f(p)(n+1)^c$.
We represent graphs by adjacency matrices and vertex subsets by bit
vectors; all logarithms are base two.

\Needspace{9\baselineskip}

\section{Constrained search on unbreakable graphs}\label{sec:local-algorithm}

In this section, we assume that the input graph $G$ is
$(q,k)$-unbreakable for given integers $q,k\ge1$. Given constraint sets
$I,O\subseteq\nodes(G)$, our goal is to find a PMC $\Omega$ of $G$
satisfying
\begin{equation}\label{eq:task}
 1\le |\Omega|\le k,\qquad I\subseteq\Omega,\qquad
 \Omega\cap O=\varnothing,
\end{equation}
or determine that no such PMC exists.

\begin{theorem}\label{thm:main}
Let $q,k$ be strictly positive integers, and let $G$ be a $(q,k)$-unbreakable graph. Let $I,O\subseteq \nodes(G)$.
There is a deterministic algorithm that returns a PMC $\Omega \in \PMC_k(G,I,O)$, or determines that this set is empty, in time
\begin{equation}\label{eq:main-time}
 2^{O(k\log(2+q/k))}(n+1)^3.
\end{equation}
The algorithm is fixed-parameter tractable in $(q,k)$, and in $k$ whenever
$q$ is bounded by a computable function of $k$. In particular,
$q=O(k)$ gives time $2^{O(k)}(n+1)^3$.
\end{theorem}

The key structural property underlying Theorem~\ref{thm:main} is that
whenever $n>3q+k$, every PMC $\Omega$ of $G$ of size at most $k$
satisfies $S\subsetneq\Omega\subseteq C\cup S$, where:
\begin{enumerate}
	\item $S\in\minlsep{k-1}{G}$ is a minimal separator of $G$ of size at most $k-1$;
	\item $C$ is a full connected component of $G\sminus S$ (so $N_G(C)=S$); and
	\item $|C|\leq q$.
\end{enumerate}
In other words, removing the minimal separator $S\subsetneq\Omega$
isolates the remaining vertices $\Omega{\setminus}S$ of the PMC inside a
single connected component $C$ of size at most $q$. Once such a pair
$(C,S)$ is identified, all vertices of $S$ must belong to $\Omega$, and
we need only guess the nonempty remainder $X=\Omega{\setminus}S\subseteq C$
of size at most $k-|S|$.

To generate all candidate pairs $(C,S)$, it suffices to enumerate, for
each root vertex $s\in\nodes(G)$, all connected subsets $C$ containing
$s$ such that $|C|\le q$ and $|N_G(C)|\le k-1$. Starting from $\{s\}$,
our branching procedure repeatedly picks the first undecided neighbor of
the current connected set and branches on whether to include it in $C$
(at most $q-1$ times) or exclude it into $N_G(C)$ (at most $k-1$ times).
Although not every enumerated pair $(C,N_G(C))$ has $N_G(C)$ as a
minimal separator, the structural property guarantees that every PMC of
size at most $k$ is generated from at least one enumerated pair. Each
resulting candidate $\Omega=S\cup X$ is then tested directly against
Lemma~\ref{thm:pmc} and~\eqref{eq:task}.

Algorithm~\ref{alg:local} ($\Local$) formalizes this search. The
subroutine $\Regions(G,s,q,\ell)$, presented in
Section~\ref{sec:regions}, outputs all pairs $(C,N_G(C))$ such that
$s\in C$, $G[C]$ is connected, $|C|\le q$, and $|N_G(C)|\le\ell$. The
subroutine $\CheckPMC(G,\Omega,k,I,O)$, described in
Section~\ref{sec:recognition}, tests whether $\Omega$ is a PMC of $G$
satisfying~\eqref{eq:task}.

\begin{algorithm}[H]
\small
\caption{$\Local$: one constrained PMC in a $(q,k)$-unbreakable graph}
\label{alg:local}
\begin{algorithmic}[1]
\Procedure{Local}{$G,k,I,O,q$}
 \State $n\gets |\nodes(G)|$
 \If{$n=0$ or $k\le0$ or $I\cap O\ne\varnothing$ or $|I|>k$}
  \State \Return \texttt{No}
 \EndIf
 \If{$n\le3q+k$}
  \For{each nonempty $\Omega\subseteq \nodes(G)$ with $|\Omega|\le k$}
   \If{$\CheckPMC(G,\Omega,k,I,O)=\texttt{Yes}$}
    \State \Return $\Omega$
   \EndIf
  \EndFor
  \State \Return \texttt{No}
 \EndIf
 \For{each $s\in \nodes(G)$}
  \For{each $(C,S)$ emitted by $\Regions(G,s,q,k-1)$}
   \For{each nonempty $X\subseteq C$ with $|X|\le k-|S|$}
    \State $\Omega\gets S\cup X$
    \If{$\CheckPMC(G,\Omega,k,I,O)=\texttt{Yes}$}
     \State \Return $\Omega$
    \EndIf
   \EndFor
  \EndFor
 \EndFor
 \State \Return \texttt{No}
\EndProcedure
\end{algorithmic}
\end{algorithm}

We establish the correctness of $\Local$ in
Section~\ref{sec:correctness}, present and analyze $\Regions$ in
Section~\ref{sec:regions}, and complete the runtime analysis of
Theorem~\ref{thm:main} in Section~\ref{sec:proof}.

\section{Correctness of \Local}\label{sec:correctness}\label{sec:region}

We first prove the structural localization lemma for $(q,k)$-unbreakable
graphs with $n>3q+k$, and then deduce the correctness of
Algorithm~\ref{alg:local} assuming that $\Regions$ and $\CheckPMC$ meet
their specifications.

\begin{lemma}\label{lem:region}
Let $G$ be $(q,k)$-unbreakable, with $q,k\ge1$ and $n>3q+k$.
For every PMC $\Omega$ of size at most $k$, there is a connected set
$C$ such that, with $S=N_G(C)$,
\begin{equation}\label{eq:region}
 |C|\le q,\qquad |S|\le k-1,\qquad
 S\subsetneq\Omega\subseteq C\cup S.
\end{equation}
\end{lemma}
\begin{proof}
Since $|\Omega|\leq k$ and $n>3q+k$, the graph $G\sminus\Omega$ contains
more than $3q$ vertices. We claim that $G\sminus\Omega$ has a connected
component $D$ of size greater than $q$. Suppose, by way of
contradiction, that every component of $G\sminus\Omega$ has at most $q$
vertices. Accumulating components of $G\sminus\Omega$ one by one until
their union first exceeds $q$ yields a union of components of size at
most $2q$, and the remaining components of $G\sminus\Omega$ form a
disjoint union of size greater than $3q-2q=q$. Since $\Omega$ separates
these two unions and $|\Omega|\le k$, this contradicts the assumption
that $G$ is $(q,k)$-unbreakable.

Let $S=N_G(D)$. Because $D$ is a connected component of
$G\sminus\Omega$, we have $S\subseteq\Omega$, and by
Lemma~\ref{thm:pmc}(1), $D$ cannot be full with respect to $\Omega$.
Hence $S\subsetneq\Omega$, so $\Omega{\setminus}S\ne\emptyset$. Pick any
vertex $s\in\Omega{\setminus}S$ and let $C\eqdef C_s(G\sminus S)$ be the
connected component of $G\sminus S$ containing $s$.

We first show that $\Omega{\setminus}S\subseteq C$. Take any
$u\in\Omega{\setminus}(S\cup\set{s})$. By the pair condition of
Lemma~\ref{thm:pmc}, $G$ contains an $s,u$-path whose internal vertices
lie outside $\Omega$. Neither the endpoints $s,u$ nor the internal
vertices of this path belong to $S\subsetneq\Omega$, so the entire
$s,u$-path resides in $G\sminus S$. Consequently, $u\in C_s(G\sminus S)=C$,
and thus $\Omega{\setminus}S\subseteq C$, which together with
$S\subsetneq\Omega$ yields $S\subsetneq\Omega\subseteq C\cup S$.

We next show that $N_G(C)=S$. Since $C$ is a connected component of
$G\sminus S$, we have $N_G(C)\subseteq S$. For the reverse inclusion,
take any $u\in S\subsetneq\Omega$. By the pair condition of
Lemma~\ref{thm:pmc}, there is an $s,u$-path in $G$ whose internal
vertices avoid $\Omega$ and hence avoid $S$. Removing $u$ from this path
leaves a connected path in $G\sminus S$ starting at $s$, so the
predecessor of $u$ along the path belongs to $C_s(G\sminus S)=C$. Thus
$u\in N_G(C)$, which establishes $S\subseteq N_G(C)$ and hence
$N_G(C)=S$.

Finally, observe that $D$ is also a connected component of $G\sminus S$
with $N_G(D)=S$. Because $s\in C\cap\Omega$ and $D\cap\Omega=\emptyset$,
the components $C$ and $D$ are disjoint, so
$D\subseteq\nodes(G){\setminus}(C\cup S)$. In the separation
\[
 \bigl(C,S,\nodes(G){\setminus}(C\cup S)\bigr),
\]
the separator $S\subsetneq\Omega$ has size at most $|\Omega|-1\le k-1$,
and the side $\nodes(G){\setminus}(C\cup S)$ contains $D$ and thus has
more than $q$ vertices. By the $(q,k)$-unbreakability of $G$, we conclude that $|C|\le q$. This completes the proof.
\end{proof}

As a simple illustration of Lemma~\ref{lem:region}, consider a graph
obtained by attaching a degree-one vertex $s$ to a vertex $x$ of a
clique on at least six vertices. This graph is $(1,2)$-unbreakable, as
deleting at most two vertices leaves the remaining clique connected and
can isolate at most the leaf $s$. For the PMC $\Omega=\{s,x\}$,
Lemma~\ref{lem:region} identifies the connected region $C=\{s\}$ of size
$1$ and its external neighborhood $S=N_G(C)=\{x\}$, regardless of how
large the clique is.

\begin{lemma}\label{lem:local-correctness}
Assuming that $\Regions$ satisfies Lemma~\ref{lem:regions}, for every
$(q,k)$-unbreakable graph $G$ with $q,k\ge1$ and sets
$I,O\subseteq \nodes(G)$, Algorithm~\ref{alg:local} terminates and
returns a member of $\PMC_k(G,I,O)$ if this family is nonempty, and
returns \texttt{No} otherwise.
\end{lemma}
\begin{proof}
If $n=0$, $k\le0$, $I\cap O\ne\varnothing$, or $|I|>k$, then
$\PMC_k(G,I,O)=\emptyset$ and Algorithm~\ref{alg:local} correctly returns
\texttt{No} in line~4. In addition, every set $\Omega$ returned by the
algorithm passes $\CheckPMC(G,\Omega,k,I,O)$, which verifies both
Lemma~\ref{thm:pmc} and~\eqref{eq:task}; hence any returned set belongs
to $\PMC_k(G,I,O)$.

Conversely, suppose there exists $\Omega\in\PMC_k(G,I,O)$. If
$n\le3q+k$, the exhaustive loop in lines~7--11 inspects all nonempty
subsets of $\nodes(G)$ of size at most $k$ and therefore finds a valid
solution. If $n>3q+k$, Lemma~\ref{lem:region} guarantees a connected set
$C$ with $|C|\le q$ and $S=N_G(C)$ such that $|S|\le k-1$ and
$S\subsetneq\Omega\subseteq C\cup S$. For any root $s\in C$, the call
$\Regions(G,s,q,k-1)$ outputs the pair $(C,S)$. Since
$X\eqdef\Omega{\setminus}S$ is a nonempty subset of $C$ of size at most
$k-|S|$, the inner loop in lines~16--21 examines $S\cup X=\Omega$. Thus
Algorithm~\ref{alg:local} is guaranteed to return a member of
$\PMC_k(G,I,O)$ whenever one exists. Termination follows immediately
from the finiteness of the subset loops and the termination of
$\Regions$ and $\CheckPMC$.
\end{proof}

\section{Enumerating the connected regions}\label{sec:regions}

Given a root vertex $s\in\nodes(G)$, a fixed ordering of $\nodes(G)$,
and integers $q\ge1$ and $\ell\ge0$, the procedure $\Regions(G,s,q,\ell)$
(Algorithm~\ref{alg:regions}) lists all pairs $(C,N_G(C))$ such that
$s\in C$, $G[C]$ is connected, $|C|\le q$, and $|N_G(C)|\le\ell$. Note
that Algorithm~\ref{alg:regions} does not require $G$ to be unbreakable.

Each recursive call $\texttt{Visit}(P,D)$ maintains a connected subset
$P\subseteq\nodes(G)$ containing $s$ and a subset $D\subseteq N_G(P)$ of
neighbors permanently excluded from the region being grown. Initially,
$P=\{s\}$ and $D=\varnothing$. Whenever $N_G(P){\setminus}D$ is nonempty
and neither size budget is exhausted ($|P|<q$ and $|D|<\ell$), the
procedure selects the first undecided neighbor $v\in N_G(P){\setminus}D$
and branches on either adding $v$ to $P$ (which preserves the
connectivity of $P$) or adding $v$ to $D$.

If either budget is reached, no further branching is necessary: when
$|P|=q$, no additional vertices can be added to $P$, so $P$ itself is the
only possible completion (valid if and only if $|N_G(P)|\le\ell$); when
$|D|=\ell$, no further neighbors can be excluded, so any valid completion
must absorb all remaining reachable vertices of $G\sminus D$ and thus
equal the connected component $C_s(G\sminus D)$ (valid if and only if
$|C_s(G\sminus D)|\le q$). Handling these boundary cases directly allows
us to bound the size of the entire recursion tree by a function of $q$
and $\ell$.

\begin{algorithm}[H]
\small
\caption{$\Regions$: connected sets with size at most $q$ and neighborhood size at most $\ell$.}
\label{alg:regions}
\begin{algorithmic}[1]
\Procedure{Regions}{$G,s,q,\ell$}
 \State Fix an order of $\nodes(G)$
 \State \Call{Visit}{$\set{s},\varnothing$}
\EndProcedure
\Statex
\Procedure{Visit}{$P,D$}
 \State \algorithmicif\ $N_G(P){\setminus} D=\varnothing$ \algorithmicthen\ Output $(P,D)$; \Return
 \If{$|P|=q$}
  \State \algorithmicif\ $|N_G(P)|\le\ell$ \algorithmicthen\ Output $(P,N_G(P))$
  \State \Return
 \EndIf
 \If{$|D|=\ell$}
  \State \algorithmicif\ $|C_s(G\sminus D)|\le q$ \algorithmicthen\ Output $(C_s(G\sminus D),D)$
  \State \Return
 \EndIf
 \State $v\gets$ the first vertex of $N_G(P){\setminus} D$
 \State \Call{Visit}{$P\cup\set{v},D$}
 \State \Call{Visit}{$P,D\cup\set{v}$}
\EndProcedure
\end{algorithmic}
\end{algorithm}

\begin{lemma}\label{lem:regions}
For a fixed root $s$, Algorithm~\ref{alg:regions} emits exactly the
pairs $(C,N_G(C))$ with $s\in C$, $C$ connected, $|C|\le q$, and
$|N_G(C)|\le\ell$, each once. It produces at most
\begin{equation}\label{eq:region-count}
 B(q,\ell)=\binom{q+\ell-1}{\ell}
\end{equation}
outputs and makes at most $2B(q,\ell)-1$ recursive calls, running in
total time $O(B(q,\ell)(n+1)^2)$.
\end{lemma}
\begin{proof}
By a straightforward induction on the recursion depth, every call
$\texttt{Visit}(P,D)$ maintains the invariants that $s\in P$, $G[P]$ is
connected, $P\cap D=\varnothing$, $D\subseteq N_G(P)$, $|P|\le q$, and
$|D|\le\ell$.

We first verify that every pair output by $\texttt{Visit}(P,D)$ is
valid. If line~6 triggers ($N_G(P){\setminus}D=\varnothing$), the
invariant $D\subseteq N_G(P)$ implies $N_G(P)=D$, so $(P,D)=(P,N_G(P))$
satisfies $|P|\le q$ and $|N_G(P)|\le\ell$. If line~8 triggers
($|P|=q$), the explicit check $|N_G(P)|\le\ell$ ensures that
$(P,N_G(P))$ is valid. If line~12 triggers ($|D|=\ell$), then because
$s\in P$, $G[P]$ is connected, and $P\cap D=\varnothing$, we have
$P\subseteq C_s(G\sminus D)$. Since $D\subseteq N_G(P)$, every vertex of
$D$ has a neighbor in $P\subseteq C_s(G\sminus D)$, which combined with
$N_G(C_s(G\sminus D))\subseteq D$ gives $N_G(C_s(G\sminus D))=D$; the
check $|C_s(G\sminus D)|\le q$ then confirms that $(C_s(G\sminus D),D)$
is valid.

Conversely, let $C_*$ be any connected subset of $\nodes(G)$ containing
$s$ with $|C_*|\le q$ and $|N_G(C_*)|\le\ell$. We trace the unique
root-to-leaf path in the recursion tree that follows line~16 when
$v\in C_*$ and line~17 when $v\notin C_*$. Along this path, we
inductively maintain $P\subseteq C_*$ and $D\subseteq N_G(C_*)$, because
every branched vertex $v\in N_G(P){\setminus}D$ either belongs to $C_*$
or has a neighbor in $P\subseteq C_*$ and thus belongs to $N_G(C_*)$.
When this path reaches a leaf, one of the three stopping conditions
holds:
(i)~if $N_G(P){\setminus}D=\varnothing$, then $N_G(P)=D\subseteq N_G(C_*)$
is disjoint from $C_*$, so the connectivity of $G[C_*]$ forces $P=C_*$;
(ii)~if $|P|=q$, then $P\subseteq C_*$ and $|C_*|\le q$ force $P=C_*$;
(iii)~if $|D|=\ell$, then $D\subseteq N_G(C_*)$ and $|N_G(C_*)|\le\ell$
force $D=N_G(C_*)$, and since $G[C_*]$ is a connected subgraph of
$G\sminus D$ whose external neighborhood is $D$, $C_*$ must be the entire
component $C_s(G\sminus D)$. In all three cases, the leaf outputs
$(C_*,N_G(C_*))$. Because any two distinct root-to-leaf paths disagree on
the membership in $C$ of their lowest common ancestor's branched vertex
$v$, no pair $(C_*,N_G(C_*))$ is emitted more than once.

To bound the number of leaves and recursive calls, let $a=q-|P|$ and
$b=\ell-|D|$ denote the remaining budgets at a call
$\texttt{Visit}(P,D)$. If $a=0$ or $b=0$, the call is a leaf. If $a,b\ge1$,
the call either terminates immediately at line~6 or makes two recursive
calls with remaining budgets $(a-1,b)$ and $(a,b-1)$. By induction on
$a+b$ using Pascal's identity
\[
 \binom{a+b-1}{b}+\binom{a+b-1}{b-1}=\binom{a+b}{b},
\]
a subtree rooted at a state with budgets $(a,b)$ has at most
$\binom{a+b}{b}$ leaves. At the initial call $\texttt{Visit}(\{s\},\varnothing)$,
we have $a=q-1$ and $b=\ell$, which yields at most $B(q,\ell)$ leaves
as stated in~\eqref{eq:region-count} and at most $2B(q,\ell)-1$ calls in
total. Since each call performs $O((n+1)^2)$ work to update neighborhoods
or compute $C_s(G\sminus D)$ via the adjacency matrix, the total running
time is $O(B(q,\ell)(n+1)^2)$.
\end{proof}

A similar bound on the number of connected subsets with small
neighborhoods was given by Fomin and
Villanger~\cite[Section~3]{fomin2008}; Lemma~\ref{lem:regions} also
covers the boundary case $\ell=0$, in which the only candidate is the
connected component $C_s(G)$ itself (output if and only if
$|C_s(G)|\le q$). Note that in Algorithm~\ref{alg:local}, pairs $(C,S)$
can be consumed on the fly during the depth-first traversal of
$\Regions$ without storing the full list in memory, so the recursion uses
$O((n+1)^2)$ space.

\section{Total running time}\label{sec:proof}

We now detail the implementation of $\CheckPMC$ and combine its cost
with Lemma~\ref{lem:regions} to complete the proof of
Theorem~\ref{thm:main}.

\subsection{Checking a candidate in the whole graph}\label{sec:recognition}

Given $(G,\Omega,k,I,O)$, the subroutine $\CheckPMC$ first verifies that
$1\le|\Omega|\le k$, $I\subseteq\Omega$, and $\Omega\cap O=\varnothing$.
It then computes the connected components of $G\sminus\Omega$ and their
external neighborhoods in $\Omega$. By Lemma~\ref{thm:pmc}, if any
component $D$ of $G\sminus\Omega$ has $N_G(D)=\Omega$, $\CheckPMC$
returns $\texttt{No}$. Otherwise, it initializes a $|\Omega|\times|\Omega|$
boolean table marking all pairs $u,v\in\Omega$ that are adjacent in $G$,
and for each component $D$ of $G\sminus\Omega$, marks all pairs
$u,v\in N_G(D)$. It returns $\texttt{Yes}$ if and only if every pair of
distinct vertices in $\Omega$ is marked.

Using an adjacency matrix, finding the connected components of
$G\sminus\Omega$ and their neighborhoods in $\Omega$ takes $O(n^2)$
time. Because $|\Omega|\le k$ and $G\sminus\Omega$ has at most $n$
components, updating the pairwise table across all components takes
$O(nk^2)$ time. Thus, a single call to $\CheckPMC$ runs in time
\begin{equation}\label{eq:recognition-time}
 O\bigl(n^2+nk^2\bigr)
 \ \subseteq\ O\bigl(k^2n^2\bigr).
\end{equation}
Notice that $\CheckPMC$ always evaluates components and neighborhoods in
the full graph $G\sminus\Omega$, and tests the constraints $I,O$ on the
entire candidate $\Omega$.

\subsection{Time for the complete search}

\begin{proof}[Running-time bound in Theorem~\ref{thm:main}]
After the initial trivial checks in Algorithm~\ref{alg:local}, we have
$n,k\ge1$, and each candidate check via $\CheckPMC$ takes $O(k^2n^2)$
time by~\eqref{eq:recognition-time}. Set
\[
 A=\sum_{i=1}^{k}\binom{3q+k}{i},\qquad
 B=\binom{q+k-2}{k-1},\qquad
 H=\sum_{i=1}^{\min\{q,k\}}\binom{q}{i}.
\]
When $n\leq 3q+k$, Algorithm~\ref{alg:local} tests at most $A$ subsets
in lines~7--11 in $O(k^2An^2)$ total time.

When $n>3q+k$, Algorithm~\ref{alg:local} invokes
$\Regions(G,s,q,k-1)$ for each of the $n$ roots $s\in\nodes(G)$. By
Lemma~\ref{lem:regions}, each call to $\Regions$ runs in $O(Bn^2)$ time
and emits at most $B$ pairs $(C,S)$. For each emitted pair $(C,S)$,
since $|C|\le q$ and $k-|S|\le k$, the inner loop tests at most $H$
nonempty subsets $X\subseteq C$. Across all $n$ roots, the algorithm
performs at most $nBH$ calls to $\CheckPMC$ (costing $O(k^2BHn^3)$ time)
and spends $O(Bn^3)$ time inside $\Regions$. Since $k,H\ge1$, the
overall running time of Algorithm~\ref{alg:local} is bounded by
\begin{equation}\label{eq:explicit-time}
 \begin{cases}
  O(k^2An^2), & n\le3q+k,\\
  O(k^2BHn^3), & n>3q+k.
 \end{cases}
\end{equation}

Applying the standard binomial sum inequality
$\sum_{i=0}^{\min\{u,k\}}\binom{u}{i}\le\binom{u+k}{k}$
(which follows from Vandermonde's identity~\cite[Eq.~(1.10), p.~3]{johnson2005}),
we have
\[
 A\le\binom{3q+2k}{k},\qquad
 H\le\binom{q+k}{k},\qquad
 B\le\binom{q+k}{k}.
\]
Substituting these bounds into~\eqref{eq:explicit-time} shows that in
both regimes the running time is at most
\begin{equation}\label{eq:binomial-time}
 O\!\left(k^2\left[
 \binom{3q+2k}{k}n^2+\binom{q+k}{k}^{\!2}n^3
 \right]\right).
\end{equation}
Since $\binom{u}{k}\le(eu/k)^k$ for $u\ge k\ge1$~\cite[Proposition~1.4]{jukna2011},
expression~\eqref{eq:binomial-time} simplifies to
$2^{O(k\log(2+q/k))}n^3$. Combining this bound with
Lemmas~\ref{lem:local-correctness} and~\ref{lem:regions} completes the
proof of Theorem~\ref{thm:main}.
\end{proof}

\section{Choosing \texorpdfstring{$q$}{q}}\label{sec:q}

While Algorithm $\Local$ is valid for any integer $q\ge1$ such that $G$
is $(q,k)$-unbreakable, its running time in Theorem~\ref{thm:main} grows
with $q$. Because $(q,k)$-unbreakability implies $(q',k)$-unbreakability
for all $q'\ge q$ but not for $q'<q$, Theorem~\ref{thm:main} is most
effective when invoked with the smallest threshold $q$ for which $G$ can
be certified $(q,k)$-unbreakable.

Indeed, passing a threshold $q$ smaller than the true unbreakability of
$G$ can cause $\Local$ to miss valid solutions. To illustrate why the
unbreakability promise cannot in general be lowered, consider a graph
$G$ formed by the disjoint union of two cliques of size $Q\ge2$.
Because $|\nodes(G)|=2Q$, no separation can leave more than $Q$ vertices
on both sides, so $G$ is $(Q,k)$-unbreakable. However, the empty
separator $Z=\emptyset$ partitions $G$ into two sides of size $Q>Q-1$,
so $G$ is $(Q-1,k)$-breakable. In particular, whenever $G$ is
$(q,k)$-unbreakable for some $q=O(k)$, Theorem~\ref{thm:main} yields a
running time of $2^{O(k)}n^3$.

\section{When the graph is \texorpdfstring{$(q,k)$}{(q,k)}-breakable}\label{sec:general}

When $G$ is $(q,k)$-breakable, it admits a separation $(L,Z,U)$ with
$|Z|\le k$, $|L|>q$, and $|U|>q$, decomposing $G$ as
$G=G[L\cup Z]\oplus_Z G[U\cup Z]$. Following the recursive understanding
paradigm of Chitnis et al.~\cite{chitnis2016} and its adaptation to
unbreakable reductions by Lokshtanov et al.~\cite{lokshtanov2018}, we
wish to replace one piece—say $H=G[L\cup Z]$—with a smaller graph $K$
whose size is bounded solely in terms of $k$, attached to
$J=G[U\cup Z]$ along the same boundary $Z$. If replacing $H$ by $K$
preserves the existence of a constrained PMC of size at most $k$ and
$|\nodes(K)|<|\nodes(H)|$, then $K\oplus_Z G[U\cup Z]$ is an equivalent
instance on strictly fewer vertices.

The main technical challenge in designing such a replacement is that a
PMC $\Omega\in\PMC_k(G,I,O)$ can straddle both sides of the separation
($\Omega\cap L\ne\emptyset$ and $\Omega\cap U\ne\emptyset$), and its
validity depends on connected components of $G\sminus\Omega$ that may
cross $Z$. Rather than tailoring $K$ to the specific outside subgraph
$G[U\cup Z]$, we construct a universal bounded replacement $(K,I_K,O_K)$
for $(H,I_H,O_H)$ and boundary $Z$: for \emph{every} outside graph $J$
satisfying $\nodes(H)\cap\nodes(J)=\nodes(K)\cap\nodes(J)=Z$ and every
choice of constraints $I_J,O_J\subseteq\nodes(J)$, the combined graph
$H\oplus_Z J$ admits a constrained PMC of size at most $k$ if and only
if $K\oplus_Z J$ does.

We break the construction of $(K,I_K,O_K)$ into two steps:
\begin{enumerate}
\item \textbf{Compressing $H$ for a fixed retained set $R$ (Section~\ref{sec:compression}):}
Suppose we have already selected a subset $R\subseteq\nodes(H)$ of
retained vertices such that $Z\cup I_H\subseteq R$. We construct a
graph $K$ on $R$ plus at most $|R|$ fresh auxiliary vertices (all placed
in $O_K$) such that $|\nodes(K)|\le 2|R|$ and, for every outside graph
$J$ attached along $Z$ and every candidate
$\Omega\subseteq R\cup\nodes(J)$, $\Omega$ is a constrained PMC of size
at most $k$ in $H\oplus_Z J$ if and only if $\Omega$ is a constrained
PMC of size at most $k$ in $K\oplus_Z J$.

\item \textbf{Selecting a small retained set $R$ via test graphs (Section~\ref{sec:retained}):}
Compression with an arbitrary set $R$ preserves only those solutions
$\Omega$ whose intersection with $\nodes(H)$ happens to lie in $R$. When
$H$ is $(q,k)$-unbreakable, we show how to compute a set $R$ of size
$2^{O(k^2)}$ that is guaranteed to contain $\Omega'\cap\nodes(H)$ for at
least one solution $\Omega'$ whenever $H\oplus_Z J$ admits a solution.
To do so without inspecting $J$, we enumerate a finite family of test
graphs $T$, each consisting of $Z$ and at most $k+1$ fresh vertices,
together with all subsets $I_T\subseteq\nodes(T)$ of size at most $k$
and $O_T=\nodes(T){\setminus}I_T$. Because each $H\oplus_Z T$ is
$(q+k+1,k)$-unbreakable, we can invoke $\Local$ on every test instance
and define $R$ as the union of $Z\cup I_H$ with the intersections
$\Omega\cap\nodes(H)$ of all returned witnesses.
\end{enumerate}

Finally, in Section~\ref{sec:general-algorithm}, we combine these two
steps with the cut-finding algorithm of Chitnis et
al.~\cite[Section~3.5]{chitnis2016} (in the form of Golovach et
al.~\cite[Lemma~2, p.~18:3]{golovach2018}) to handle pieces that are themselves breakable and establish Theorem~\ref{thm:general-extend}.
\section{Compressing a piece after fixing the retained vertices}\label{sec:compression}

We now formalize the first step of our replacement scheme. Given an
arbitrary graph $H$ (which need not be unbreakable), a boundary
$Z\subseteq\nodes(H)$, and a set $R\subseteq\nodes(H)$ of retained
vertices with $Z\subseteq R$, we construct a compressed graph $K$ such
that $\nodes(K)\cap\nodes(H)=R$, $|\nodes(K)|\le 2|R|$, and all vertices
in $\nodes(K){\setminus}R$ are fresh auxiliary vertices excluded from any
solution. Theorem~\ref{thm:compression} shows that $K$ preserves PMC
membership for every candidate $\Omega\subseteq R\cup\nodes(J)$ under
any outside attachment $J$ along $Z$, and
Corollary~\ref{cor:compression-constraints} records the resulting
equivalence for constrained PMC search.

\def\mainCompressionThm#1{
	Given a graph $H$ and sets $Z\subseteq R\subseteq\nodes(H)$,
	we can construct a graph $K$ such that
	\[
	\nodes(K)\cap\nodes(H)=R,
	\qquad |\nodes(K)|\le 2|R|.
	\]
	For every graph $J$ sharing exactly $Z$ with both $H$ and $K$,
	and every nonempty set $\Omega\subseteq R\cup\nodes(J)$,
	\begin{equation}#1{eq:compression-pmc}
		\Omega\in\PMC(H\oplus_Z J)
		\quad\Longleftrightarrow\quad
		\Omega\in\PMC(K\oplus_Z J).
	\end{equation}
	The construction takes $O(n^3)$ time, where $n=|\nodes(H)|\ge1$.
}
\begin{theorem}[Compression after retaining $R$]\label{thm:compression}
\mainCompressionThm{\label}
\end{theorem}

Note that in~\eqref{eq:compression-pmc}, the same vertex set $\Omega$ is
tested in both $H\oplus_Z J$ and $K\oplus_Z J$: because
$\Omega\subseteq R\cup\nodes(J)$, $\Omega$ contains neither deleted
vertices from $\nodes(H){\setminus}R$ nor auxiliary vertices from
$\nodes(K){\setminus}R$.

When $H$ and $J$ carry inclusion and exclusion constraints $(I_H,O_H)$
and $(I_J,O_J)$ with $Z\cup I_H\subseteq R$, excluding the auxiliary
vertices $\nodes(K){\setminus}R$ via $O_K$ immediately extends
Theorem~\ref{thm:compression} to constrained PMCs of size at most $k$.

\begin{corollary}[Preserving constrained candidates]
	\label{cor:compression-constraints}
	Let $H$ be a graph, let $I_H,O_H\subseteq\nodes(H)$, and let
	$Z,R\subseteq\nodes(H)$ satisfy $Z\cup I_H\subseteq R\subseteq\nodes(H)$.
	We can construct a graph $K$ such that
	\[
	\nodes(K)\cap\nodes(H)=R,
	\qquad |\nodes(K)|\le 2|R|,
	\]
	with inclusion and exclusion constraints
	\[
	I_K\eqdef I_H,
	\qquad
	O_K\eqdef (O_H\cap R)\cup\bigl(\nodes(K){\setminus} R\bigr)
	\]
	such that for every graph $J$ satisfying
	\[
	\nodes(J)\cap\nodes(H)
	=
	\nodes(J)\cap\nodes(K)
	=
	Z,
	\]
	every $I_J,O_J\subseteq\nodes(J)$, every integer $k\ge1$,
	and every candidate $\Omega\subseteq R\cup\nodes(J)$,
	\begin{align}\label{eq:compression-constrained}
			\Omega\in\PMC_k(H\oplus_Z J,I_H\cup I_J,O_H\cup O_J)
			&&\Longleftrightarrow&&
			\Omega\in\PMC_k(K\oplus_Z J,I_K\cup I_J,O_K\cup O_J).
	\end{align}
	The construction takes $O(n^3)$ time, where
	$n=|\nodes(H)|\ge1$.
\end{corollary}
\begin{proof}
	Let $K$ be the graph constructed in Theorem~\ref{thm:compression}, and
	define $I_K$ and $O_K$ as above. The case $\Omega=\emptyset$ is immediate
	since PMCs are nonempty. For any fixed $J,I_J,O_J,k$ and any
	nonempty candidate $\Omega\subseteq R\cup\nodes(J)$,
	Theorem~\ref{thm:compression} states that $\Omega\in\PMC(H\oplus_Z J)$
	if and only if $\Omega\in\PMC(K\oplus_Z J)$. Since $I_K=I_H$, the
	inclusion constraint $I_H\cup I_J\subseteq\Omega$ is identical on both
	sides, as is the cardinality bound $|\Omega|\le k$. For the exclusion
	constraints, because $\Omega\subseteq R\cup\nodes(J)$ is disjoint from
	both $\nodes(H){\setminus}R$ and $\nodes(K){\setminus}R$, we have
	\[
	\Omega\cap(O_H\cup O_J)
	= \Omega\cap\bigl((O_H\cap R)\cup O_J\bigr)
	= \Omega\cap(O_K\cup O_J).
	\]
	Hence $\Omega\cap(O_H\cup O_J)=\emptyset$ if and only if
	$\Omega\cap(O_K\cup O_J)=\emptyset$, which establishes~\eqref{eq:compression-constrained}.
	Notice also that because $\nodes(K){\setminus}R\subseteq O_K$, every
	set $\Omega\in\PMC_k(K\oplus_Z J,I_K\cup I_J,O_K\cup O_J)$ automatically
	satisfies $\Omega\subseteq R\cup\nodes(J)$, so the reverse direction
	of~\eqref{eq:compression-constrained} applies to all constrained PMCs
	of $K\oplus_Z J$.
\end{proof}
We now prove Theorem~\ref{thm:compression} in three steps: characterizing
how the torso operation affects PMC membership (Section~\ref{sec:torso}),
restoring lost full-component obstructions with forbidden vertices
(Section~\ref{sec:forbidden-vertex}), and bounding the number of added
vertices by $|R|$ (Section~\ref{sec:compression-size}).
\subsection{PMC membership in the torso}\label{sec:torso}
Let $H$ be a graph and $R\subseteq \nodes(H)$.
For each connected component $C$ of $H\sminus R$, its external
neighborhood $S_C=N_H(C)$ is a subset of $R$. Recall from
Section~\ref{sec:preliminaries} that the torso $F=\Tor(H,R)$ is formed
from $H[R]$ by making each $S_C$ a clique~\cite[Section~2.3, Definition~2.5]{marx2013}.
Replacing each maximal subpath through a deleted component $C$ of
$H\sminus R$ with the torso edge between its entry and exit vertices in
$S_C$ (and vice versa) preserves the pairwise connectivity condition of
Lemma~\ref{thm:pmc}. However, deleting $C$ can eliminate a full component:
in the path $a$--$c$--$b$ with $R=\{a,b\}$, the deleted component
$C=\{c\}$ has $N_H(C)=\{a,b\}$, so $\{a,b\}$ is not a PMC of the path
even though it becomes a PMC of the torso edge $\Tor(H,R)$. The
following lemma shows that this is the only way taking the torso can
alter PMC membership.

\begin{lemma}[PMC membership in the torso]\label{lem:torso}
Let $Z\subseteq R\subseteq\nodes(H)$ and $F=\Tor(H,R)$.
For every graph $J$ sharing exactly $Z$ with $H$, and every
nonempty $\Omega\subseteq R\cup\nodes(J)$,
\begin{equation}\label{eq:torso-error}
 \Omega\in\PMC(H\oplus_Z J)
 \quad\Longleftrightarrow\quad
 \begin{cases}
 \Omega\in\PMC(F\oplus_Z J),\text{ and}\\
 \Omega\ne N_H(C)\text{ for every component }C\text{ of }H\sminus R.
 \end{cases}
\end{equation}
\end{lemma}
\begin{proof}
	Let $G_H\eqdef H\oplus_Z J$ and $G_F\eqdef F\oplus_Z J$, so that
	$\nodes(G_F)=R\cup\nodes(J)$ and
	$\nodes(G_H){\setminus}\nodes(G_F)=\nodes(H){\setminus} R$.
	Because $\nodes(H)\cap\nodes(J)=Z\subseteq R$, every connected
	component $C$ of $H\sminus R$ is disjoint from $\nodes(J)$ and has no
	incident edges in $J$; hence $N_{G_H}(C)=N_H(C)\subseteq R$, and $C$
	is disjoint from $\Omega\subseteq\nodes(G_F)$.
	
	\emph{Forward direction ($\Rightarrow$).}
	Suppose $\Omega\in\PMC(G_H)$. First, if $\Omega=N_H(C)$ for some
	component $C$ of $H\sminus R$, then since $C\cap\Omega=\emptyset$ and
	$N_{G_H}(C)=N_H(C)=\Omega$, the set $C$ is a connected component of
	$G_H\sminus\Omega$ that is full with respect to $\Omega$, which contradicts Lemma~\ref{thm:pmc}(1). Thus $\Omega\ne N_H(C)$ for every component
	$C$ of $H\sminus R$.
	
	To see that $\Omega$ satisfies the pair condition in $G_F$, take any
	distinct $a,b\in\Omega$. By Lemma~\ref{thm:pmc}(2), $G_H$ contains an
	$a,b$-path $P$ whose internal vertices avoid $\Omega$. Any maximal
	contiguous subsegment $c_1,\ldots,c_t$ of $P$ contained in
	$\nodes(G_H){\setminus}\nodes(G_F)=\nodes(H){\setminus}R$ lies inside
	a single component $C$ of $H\sminus R$, and its immediately preceding
	and succeeding vertices $u,v$ on $P$ belong to $N_{G_H}(C)=N_H(C)\subseteq R$.
	Since $N_H(C)$ is a clique in $F$, $(u,v)\in\edges(G_F)$. Replacing
	each such subsegment $u,c_1,\ldots,c_t,v$ of $P$ with the single torso
	edge $(u,v)$ yields an $a,b$-path in $G_F$ whose internal vertices are
	a subset of those of $P$ and therefore avoid $\Omega$. Thus $\Omega$
	satisfies the pair condition in $G_F$.
	
	It remains to show that $G_F\sminus\Omega$ has no full component with
	respect to $\Omega$. Suppose, by way of contradiction, that $D_0$ is a
	component of $G_F\sminus\Omega$ with $N_{G_F}(D_0)=\Omega$. Replacing
	each torso edge in $G_F[D_0]$ with a path through the component of
	$H\sminus R$ that generated it (and shortcutting any repeated vertices)
	shows that all vertices of $D_0$ belong to a single connected component
	$D$ of $G_H\sminus\Omega$. Now take any $u\in\Omega$; since
	$N_{G_F}(D_0)=\Omega$, $u$ has a neighbor $v\in D_0\subseteq D$ in
	$G_F$. If $(u,v)\in\edges(G_H)$, then $u\in N_{G_H}(D)$. Otherwise,
	$(u,v)\in\edges(G_F){\setminus}\edges(G_H)$ is a torso edge, so
	$u,v\in N_H(C)$ for some component $C$ of $H\sminus R$. Since
	$C\cap\Omega=\emptyset$ and $v\in D$ has a neighbor in $C$, the entire
	connected set $C$ is contained in $D$, and because $u\in N_H(C)$, we
	again have $u\in N_{G_H}(D)$. Thus $N_{G_H}(D)=\Omega$, which contradicts $\Omega\in\PMC(G_H)$. By Lemma~\ref{thm:pmc}, we conclude that
	$\Omega\in\PMC(G_F)$.
	
	\emph{Reverse direction ($\Leftarrow$).}
	Conversely, suppose $\Omega\in\PMC(G_F)$ and $\Omega\ne N_H(C)$ for
	every component $C$ of $H\sminus R$. For any distinct $a,b\in\Omega$,
	taking an $a,b$-path in $G_F$ whose internal vertices avoid $\Omega$,
	expanding each torso edge $e\in\edges(G_F){\setminus}\edges(G_H)$
	through the component of $H\sminus R$ that supplied $e$ (which is
	disjoint from $\Omega$), and shortcutting repeated vertices produces an
	$a,b$-path in $G_H$ whose internal vertices avoid $\Omega$. Hence
	$\Omega$ satisfies the pair condition in $G_H$.
	
	Now suppose, by way of contradiction, that $G_H\sminus\Omega$ has a
	component $D$ with $N_{G_H}(D)=\Omega$. If
	$D\subseteq\nodes(H){\setminus}R$, then $D$ is a single component $C$
	of $H\sminus R$, so $\Omega=N_{G_H}(C)=N_H(C)$, contrary to our hypothesis. Otherwise, the set $D_0\eqdef D\cap\nodes(G_F)$ is
	nonempty. Contracting maximal subsegments through $\nodes(H){\setminus}R$
	along paths in $G_H[D]$ into torso edges (and vice versa) implies that
	$D_0$ is a single connected component of $G_F\sminus\Omega$. Take any
	$u\in\Omega$; since $N_{G_H}(D)=\Omega$, $u$ has a neighbor $x\in D$ in
	$G_H$. If $x\in\nodes(G_F)$, then $x\in D_0$ and $(u,x)\in\edges(G_F)$,
	so $u\in N_{G_F}(D_0)$. If instead $x\in\nodes(H){\setminus}R$, let $C$
	be the component of $H\sminus R$ containing $x$, and fix any vertex
	$y\in D_0$. Along a path from $x$ to $y$ in $G_H[D]$, the first vertex
	$v$ outside $C$ belongs to $N_{G_H}(C)\cap D=N_H(C)\cap D_0$. Since
	$u,v\in N_H(C)$, the torso edge $(u,v)$ is present in $G_F$ with
	$v\in D_0$, so again $u\in N_{G_F}(D_0)$. Consequently, $N_{G_F}(D_0)=\Omega$, which contradicts $\Omega\in\PMC(G_F)$. Thus
	$\Omega\in\PMC(G_H)$.
\end{proof}

\subsection{Restoring an obstruction with one forbidden vertex}\label{sec:forbidden-vertex}

By Lemma~\ref{lem:torso}, the only candidates $\Omega\subseteq R\cup\nodes(J)$
that become false positives in the torso are the neighborhoods
$S=N_H(C)$ of deleted components $C$ of $H\sminus R$. Because each such
$S$ is a clique in $F=\Tor(H,R)$, we can eliminate $S$ as a candidate
PMC by attaching a single excluded vertex $z_S$ with neighborhood $S$:
deleting $S$ then leaves $\{z_S\}$ as a full component, while no other
candidate avoiding $z_S$ is affected.

\begin{lemma}[One forbidden vertex]\label{lem:marker}
Let $S$ be a clique of a graph $F$. Form $F^+$ by adding a new
vertex $z_S$ with neighborhood exactly $S$.
For every nonempty $\Omega\subseteq\nodes(F)$,
\[
 \Omega\in\PMC(F^+)
 \quad\Longleftrightarrow\quad
 \Omega\in\PMC(F)\text{ and }\Omega\ne S.
\]
\end{lemma}
\begin{proof}
	\emph{Forward direction ($\Rightarrow$).}
	Suppose $\Omega\in\PMC(F^+)$, where $\Omega\subseteq\nodes(F)$ and
	hence $z_S\notin\Omega$. If $\Omega=S$, then $\{z_S\}$ is a connected
	component of $F^+\sminus\Omega$ with $N_{F^+}(\{z_S\})=S=\Omega$, in contradiction to Lemma~\ref{thm:pmc}(1). Thus $\Omega\ne S$.
	For any distinct $a,b\in\Omega$, take an $a,b$-path in $F^+$ whose
	internal vertices avoid $\Omega$; if this path visits $z_S$, its
	neighbors $p,q$ before and after $z_S$ belong to the clique
	$S=N_{F^+}(z_S)$, so replacing $p$--$z_S$--$q$ with the edge
	$(p,q)\in\edges(F)$ yields an $a,b$-path in $F$ avoiding $\Omega$
	internally. Finally, if $F\sminus\Omega$ had a full component $D$ with
	respect to $\Omega$, the component of $F^+\sminus\Omega$ containing $D$
	would also be full with respect to $\Omega$, which contradicts $\Omega\in\PMC(F^+)$. Hence $\Omega\in\PMC(F)$.
	
	\emph{Reverse direction ($\Leftarrow$).}
	Conversely, suppose $\Omega\in\PMC(F)$ and $\Omega\ne S$. Since
	$F$ is a subgraph of $F^+$, $\Omega$ immediately satisfies the pair
	condition in $F^+$. To check the full-component condition in
	$F^+\sminus\Omega$, we consider two cases. If $S\subseteq\Omega$, then
	$S\subsetneq\Omega$ (as $\Omega\ne S$), so $\{z_S\}$ forms an isolated
	component of $F^+\sminus\Omega$ with neighborhood
	$N_{F^+}(\{z_S\})=S\subsetneq\Omega$, and all other components of
	$F^+\sminus\Omega$ coincide with those of $F\sminus\Omega$ and are not
	full. If instead $S{\setminus}\Omega\ne\emptyset$, then because
	$S{\setminus}\Omega$ is a clique, it is contained in a single component
	$D$ of $F\sminus\Omega$. In $F^+\sminus\Omega$, the vertex $z_S$ merges
	into $D$ to form the component $D\cup\{z_S\}$, while all other
	components remain unchanged. Since every vertex of $S\cap\Omega$ is
	already adjacent to $S{\setminus}\Omega\subseteq D$, we have
	\[
	N_{F^+}(D\cup\{z_S\})=N_F(D)\ne\Omega.
	\]
	Thus no component of $F^+\sminus\Omega$ is full with respect to
	$\Omega$, and $\Omega\in\PMC(F^+)$.
\end{proof}

Notice that Lemma~\ref{lem:marker} relies on $S$ being a clique so that
any path with $z_S$ as an internal vertex can be shortened using an edge
within $S$. Because
$S\subseteq R$ remains a clique in $F\oplus_Z J$ for any outside graph
$J$ attached along $Z$, Lemma~\ref{lem:marker} applies directly to
$F\oplus_Z J$ as well (see Figure~\ref{fig:compression-path}).

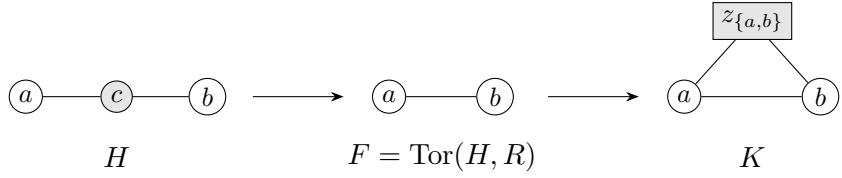
\begin{figure}[htbp]
\centering
\begin{tikzpicture}[x=1cm,y=1cm,
 vertex/.style={circle,draw,inner sep=2pt,font=\small},
 auxiliary/.style={rectangle,draw,inner sep=3pt,font=\small},
 >=Stealth]
 \node[vertex] (a1) at (0,0) {$a$};
 \node[vertex,fill=black!10] (c1) at (1.2,0) {$c$};
 \node[vertex] (b1) at (2.4,0) {$b$};
 \draw (a1)--(c1)--(b1);
 \node at (1.2,-0.8) {$H$};
 \draw[->] (3,0)--(4.2,0);
 \node[vertex] (a2) at (4.8,0) {$a$};
 \node[vertex] (b2) at (6.2,0) {$b$};
 \draw (a2)--(b2);
 \node at (5.5,-0.8) {$F=\Tor(H,R)$};
 \draw[->] (6.9,0)--(8.1,0);
 \node[vertex] (a3) at (8.7,0) {$a$};
 \node[vertex] (b3) at (10.5,0) {$b$};
 \node[auxiliary,fill=black!10] (z3) at (9.6,1) {$z_{\{a,b\}}$};
 \draw (a3)--(b3)--(z3)--(a3);
 \node at (9.6,-0.8) {$K$};
\end{tikzpicture}
\caption{Compressing the path $H=a$--$c$--$b$ with retained set
$R=\{a,b\}$. In the torso $F=\Tor(H,R)$, the pair $\{a,b\}$ becomes a
PMC; adding the excluded auxiliary vertex $z_{\{a,b\}}$ (drawn as a
rectangle) restores the full-component obstruction in $K$.}
\label{fig:compression-path}
\end{figure}

\subsection{At most one auxiliary vertex per retained vertex}\label{sec:compression-size}

Adding a separate auxiliary vertex for every connected component $C$ of
$H\sminus R$ could leave up to $|\nodes(H)|$ vertices in $K$. However,
two observations allow us to add at most $|R|$ auxiliary vertices in
total:
first, multiple components $C$ of $H\sminus R$ sharing the same
neighborhood $S=N_H(C)$ require at most a single auxiliary vertex $z_S$;
second, if $F\sminus S$ already contains a component $D$ that is full
with respect to $S$ in $F=\Tor(H,R)$ (i.e., $N_F(D)=S$), then under any
attachment $J$ along $Z$, the component $D'$ of $(F\oplus_Z J)\sminus S$
containing $D$ satisfies $N_{F\oplus_Z J}(D')=S$ and thus already
prevents $S$ from being a PMC of $F\oplus_Z J$.
Consequently, we only need to introduce an auxiliary vertex $z_S$ for
distinct nonempty neighborhoods $S=N_H(C)$ such that $F\sminus S$ has no
full component with respect to $S$. Since every such $S$ is a clique in
$F$ by construction of the torso, Lemma~\ref{thm:pmc} implies that these
sets $S$ are simultaneously cliques and PMCs of $F$.

To bound the number of such sets $S$ by $|R|=|\nodes(F)|$, we use a
known property of chordal completions~\cite[pp.~2201, 2206]{fomin2013}:
every set that is both a clique and a PMC of $F$ is a maximal clique in
\emph{every} minimal triangulation of $F$, and any $h$-vertex chordal
graph has at most $h$ maximal cliques.

\begin{citedLemma}{\cite[pp.~2201, 2206]{fomin2013}}\label{lem:clique-count}
A graph $F$ on $h$ vertices has at most $h$ nonempty sets
that are both cliques and PMCs of $F$.
\end{citedLemma}
\begin{proof}
If $h=0$, the claim is trivial. For $h\ge1$, fix any minimal
triangulation $T$ of $F$, and let $S\in\PMC(F)$ be a clique of $F$. For
each connected component $C$ of $F\sminus S$, let $Q_C=N_F(C)$; since
$S\in\PMC(F)$ is a clique, $Q_C$ is a clique satisfying
$Q_C\subsetneq S$. Let $T'$ be the subgraph of $T$ formed by the union
of $T[S]$ and the induced subgraphs $T[C\cup Q_C]$ over all components
$C$ of $F\sminus S$. Each $T[C\cup Q_C]$ is chordal and intersects the
clique $T[S]$ in the clique $Q_C$. Since gluing chordal graphs along
cliques preserves chordality and $\edges(F)\subseteq\edges(T')\subseteq\edges(T)$,
the minimality of the triangulation $T$ forces $T'=T$. In $T'=T$, every
vertex $v\in C$ has $N_T(v)\cap S\subseteq Q_C\subsetneq S$, so no
vertex outside $S$ is adjacent in $T$ to all of $S$. Thus $S$ is a
maximal clique of $T$. Because an $h$-vertex chordal graph has at most
$h$ maximal cliques~\cite{fomin2013}, there are at most $h$ such sets
$S$.
\end{proof}

\subsection{The compression algorithm and its guarantee}\label{sec:compression-algorithm}

Algorithm~\ref{alg:compress} ($\Compress$) summarizes the construction
of $(K,I_K,O_K)$ from $(H,Z,R,I_H,O_H)$.

\begin{algorithm}[htbp]
\caption{$\Compress(H,Z,R,I_H,O_H)$}\label{alg:compress}
\begin{algorithmic}[1]
\Require $Z\cup I_H\subseteq R\subseteq\nodes(H)$;
 $I_H,O_H\subseteq\nodes(H)$
\Ensure Replacement graph $K$ and its constraints $I_K,O_K$
\State Compute the components $C$ of $H\sminus R$ and their neighborhoods $S_C=N_H(C)$
\State $F\gets H[R]$
\For{each component $C$ of $H\sminus R$}
 \State Add all missing edges within $S_C$ to $F$
\EndFor
\State $K\gets F$
\For{each distinct nonempty set $S$ among the neighborhoods $S_C$}
 \If{no component $D$ of $F\sminus S$ satisfies $N_F(D)=S$}
  \State Add to $K$ a fresh vertex $z_S$ adjacent exactly to $S$
 \EndIf
\EndFor
\State $I_K\gets I_H$; $O_K\gets(O_H\cap R)\cup(\nodes(K){\setminus} R)$
\State \Return $(K,I_K,O_K)$
\end{algorithmic}
\end{algorithm}

\begin{reptheorem}{\ref{thm:compression}}
	\mainCompressionThm{\eat}
\end{reptheorem}
\begin{proof}[Proof of Theorem~\ref{thm:compression}]
Let $F=\Tor(H,R)$, and let $\mathcal{S}$ denote the family of distinct
nonempty neighborhoods $S=N_H(C)$ (over all components $C$ of
$H\sminus R$) for which $F\sminus S$ has no full component with respect
to $S$. Every $S\in\mathcal{S}$ is a clique in $F$ by construction of
the torso, and hence $S\in\PMC(F)$ by Lemma~\ref{thm:pmc}. By
Lemma~\ref{lem:clique-count}, $|\mathcal{S}|\le|R|$. Since
Algorithm~\ref{alg:compress} adds one fresh vertex $z_S$ for each
$S\in\mathcal{S}$, we have $\nodes(K)\cap\nodes(H)=R$ and
$|\nodes(K)|=|R|+|\mathcal{S}|\le 2|R|$.

\Needspace{6\baselineskip}
Now let $J$ be any graph with $\nodes(H)\cap\nodes(J)=\nodes(K)\cap\nodes(J)=Z$,
and let $\Omega\subseteq R\cup\nodes(J)$ be nonempty. Starting from
$F\oplus_Z J$, we add the auxiliary vertices $\{z_S:S\in\mathcal{S}\}$
one by one to obtain $K\oplus_Z J$. Each $S\in\mathcal{S}$ is a clique
in $F\oplus_Z J$ and remains a clique as auxiliary vertices are added,
and $z_S\notin\Omega$. Applying Lemma~\ref{lem:marker} inductively over
all $S\in\mathcal{S}$ yields
\[
\Omega\in\PMC(K\oplus_Z J)
\quad\Longleftrightarrow\quad
\Omega\in\PMC(F\oplus_Z J)
\text{ and }\Omega\notin\mathcal{S}.
\]

We claim that whenever $\Omega\in\PMC(F\oplus_Z J)$, the condition
$\Omega\notin\mathcal{S}$ is equivalent to having $\Omega\ne N_H(C)$ for
every component $C$ of $H\sminus R$. The reverse implication is
immediate since every member of $\mathcal{S}$ is of the form $N_H(C)$.
For the forward implication, suppose $\Omega\in\PMC(F\oplus_Z J)$ and
$\Omega\notin\mathcal{S}$, and let $C$ be any component of $H\sminus R$
with $S=N_H(C)$. If $S=\emptyset$, then $S\ne\Omega$ because $\Omega$ is
nonempty. If $S\ne\emptyset$ and $S\notin\mathcal{S}$, then by
definition of $\mathcal{S}$, $F\sminus S$ has a component $D$ with
$N_F(D)=S$. In $(F\oplus_Z J)\sminus S$, the component $D'$ containing
$D$ still satisfies $N_{F\oplus_Z J}(D')=S$, so $S\notin\PMC(F\oplus_Z J)$
and therefore $S\ne\Omega$. Combining this equivalence with
Lemma~\ref{lem:torso} gives
\[
\Omega\in\PMC(K\oplus_Z J)
\quad\Longleftrightarrow\quad
\Omega\in\PMC(H\oplus_Z J),
\]
establishing~\eqref{eq:compression-pmc}.

For the running time, let $n=|\nodes(H)|\ge1$. Finding the at most $n$
components of $H\sminus R$ and their neighborhoods $S_C$ takes $O(n^2)$
time, and adding the clique edges within each $S_C$ to form $F$ takes
$O(n^3)$ time in total. Deduplicating the bit vectors $S_C$ and testing,
for each distinct nonempty $S=S_C$, whether $F\sminus S$ has a full
component takes $O(n^2)$ time per neighborhood and $O(n^3)$ time overall.
Adding the at most $|R|$ auxiliary vertices and forming $(I_K,O_K)$
takes $O(n^2)$ time. Overall, the construction runs in $O(n^3)$ time.
\end{proof}

As noted in Section~\ref{sec:general}, Theorem~\ref{thm:compression}
preserves only those PMCs $\Omega$ whose intersection with $\nodes(H)$
is contained in $R$. For instance, if $H$ is the three-vertex path
$a$--$b$--$c$ with $Z=I_H=O_H=\emptyset$ and $k=2$, then
$\PMC_2(H,\emptyset,\emptyset)=\{\{a,b\},\{b,c\}\}$. If we were to choose
$R=\{b\}$, then $F=\Tor(H,R)$ would be the single vertex $b$, and
because the deleted components $\{a\}$ and $\{c\}$ both have neighborhood
$\{b\}$, $\Compress$ would add the excluded vertex $z_{\{b\}}$ and leave $\PMC_2(K,\emptyset,O_K)=\emptyset$. In the next section, we show how to
compute a set $R$ of bounded size that is guaranteed to retain at least
one solution whenever one exists.

\section{Choosing the vertices to retain}\label{sec:retained}
When the input graph $G$ is $(q,k)$-breakable, it admits a separation
$(L,Z,U)$ with $|Z|\le k$, $|L|>q$, and $|U|>q$, decomposing $G$ as
$G=H\oplus_Z J$ where
\[
H\eqdef G[L\cup Z],\qquad J\eqdef G[U\cup Z].
\]
In this section, we assume that the piece $H$ is $(q,k)$-unbreakable
(we remove this assumption in Section~\ref{sec:general-algorithm}, where
we handle pieces that are themselves breakable). Our goal is to replace
$H$ with a graph $K$ of size bounded in terms of $k$ such that
$\nodes(K)\cap\nodes(J)=Z$ and replacing $H$ by $K$ preserves the
existence of a constrained PMC of size at most $k$.

By Corollary~\ref{cor:compression-constraints}, once we fix a retained
subset $R\subseteq\nodes(H)$ containing $Z\cup(I\cap\nodes(H))$, the
procedure $\Compress$ constructs a graph $K$ with $|\nodes(K)|\le 2|R|$
such that any candidate $\Omega\subseteq R\cup\nodes(J)$ belongs to
$\PMC_k(H\oplus_Z J,I,O)$ if and only if $\Omega$ belongs to
$\PMC_k(K\oplus_Z J,I_K\cup(I\cap\nodes(J)),O_K\cup(O\cap\nodes(J)))$.
To ensure that $K\oplus_Z J$ is an equivalent instance whenever
$\PMC_k(G,I,O)\ne\emptyset$, we must choose $R$ of size bounded in terms
of $k$ so that $\PMC_k(G,I,O)$ contains at least one witness $Q$
satisfying $Q\cap\nodes(H)\subseteq R$.

In fact, we establish a universal guarantee without inspecting the
outside subgraph $J$: given $(H,Z,k,I_H,O_H,q)$, we construct a set
$R\subseteq\nodes(H)$ that works simultaneously for \emph{every} graph
$J'$ satisfying $\nodes(H)\cap\nodes(J')=Z$ and every choice of
constraints $I_{J'},O_{J'}\subseteq\nodes(J')$. Whenever
$H\oplus_Z J'$ admits a constrained PMC of size at most $k$, at least
one such PMC intersects $\nodes(H)$ solely within $R$. Combined with
$\Compress$, the resulting graph $K$ can replace $H$ inside any host
graph attached along $Z$.

The key observation behind our construction of $R$ is that although an
outside graph $J'$ may be arbitrarily large, any PMC $\Omega$ of size at
most $k$ in $H\oplus_Z J'$ selects a subset
$X=\Omega\cap\nodes(J')$ of at most $k$ vertices from $J'$. As we show
in Lemma~\ref{lem:retained-small}, for any fixed $X\subseteq\nodes(J')$
of size at most $k$, the torso $\Tor(J',Z\cup X)$ augmented with at most
one excluded auxiliary vertex yields a graph $J'_X$ on at most
$|Z|+k+1$ vertices such that a subset of $\nodes(H)$ extends $X$ to a
PMC in $H\oplus_Z J'$ if and only if it extends $X$ to a PMC in
$H\oplus_Z J'_X$.

We can therefore enumerate, up to vertex naming outside $H$, all simple
test graphs $T$ consisting of $Z$ and at most $k+1$ fresh vertices,
together with all subsets $I_T\subseteq\nodes(T)$ of size at most $k$
and $O_T=\nodes(T){\setminus}I_T$. Because $H$ is $(q,k)$-unbreakable
and $T$ adds at most $k+1$ vertices outside $H$, each attached graph
$H\oplus_Z T$ is $(q+k+1,k)$-unbreakable. We invoke $\Local$ with
threshold $q+k+1$ on $(H\oplus_Z T,k,I_H\cup I_T,O_H\cup O_T)$ for every
such test triple and collect the intersections with $\nodes(H)$ of all
returned PMCs into $R$. If $H\oplus_Z J'$ admits a constrained PMC
$\Omega$ with $\Omega\cap\nodes(J')=X$, the test triple corresponding
to $(J'_X,X,\nodes(J'_X){\setminus}X)$ succeeds, and replacing the test
vertices in its returned witness $\Omega_{T,I_T}$ with $X{\setminus}Z$
produces a valid constrained PMC $Q$ of $H\oplus_Z J'$ satisfying
$Q\cap\nodes(H)=\Omega_{T,I_T}\cap\nodes(H)\subseteq R$.

Theorem~\ref{thm:retained} states the formal guarantee and runtime bound
for constructing $R$. We keep $b=|Z|$ explicit in the statement because
our recursive reduction in Section~\ref{sec:general-algorithm} applies
this theorem to boundaries of size $b\le 2k$.

\def\thmRetained#1{
	Let $q,k\ge1$ be integers, let $H$ be a $(q,k)$-unbreakable graph with
	$n=|\nodes(H)|\ge1$, and let $Z,I_H,O_H\subseteq\nodes(H)$,
	where $|I_H|\le k$. Put
	\[
	b=|Z|,\qquad m=b+k+1,\qquad
	\tau_{b,k}=(k+2)2^{\binom m2+m}.
	\]
	There is a deterministic algorithm that returns a set
	$Z\cup I_H\subseteq R\subseteq\nodes(H)$ with
	$|R|\le b+k\tau_{b,k}$ and the following guarantee.
	For every graph $J$ with $\nodes(H)\cap\nodes(J)=Z$ and
	every $I_J,O_J\subseteq\nodes(J)$, if
	\[
	\PMC_k(H\oplus_Z J,I_H\cup I_J,O_H\cup O_J)\ne\emptyset,
	\]
	then there exists
	\begin{equation}#1{eq:retained-guarantee}
		\Omega\in\PMC_k(H\oplus_Z J,I_H\cup I_J,O_H\cup O_J)
		\quad\text{such that}\quad
		\Omega\cap\nodes(H)\subseteq R.
	\end{equation}
	The running time is
	\begin{equation}#1{eq:retained-time}
		2^{O((b+k)^2+k\log(2+q/k))}n^3.
	\end{equation}
	In particular, if $b=O(k)$, then $|R|\le2^{O(k^2)}$ and
	the running time is $2^{O(k^2+k\log(2+q/k))}n^3$.
}
\begin{theorem}[Choosing the retained vertices]\label{thm:retained}
\thmRetained{\label}
\end{theorem}

Note that if $|I_H|>k$, no attachment can admit a constrained PMC of
size at most $k$; if $H$ is empty, then $Z=I_H=\emptyset$ and
$R=\emptyset$ trivially satisfies Theorem~\ref{thm:retained}.

\subsection{A small attachment after fixing the selected outside vertices}
\label{sec:retained-small}

Let $H$ and $J$ be graphs with $\nodes(H)\cap\nodes(J)=Z$, and fix a
subset $X\subseteq\nodes(J)$ representing the vertices of $J$ selected
by a candidate $\Omega\subseteq\nodes(H)\cup\nodes(J)$ (so
$\Omega\cap\nodes(J)=X$, which in particular fixes $\Omega\cap Z=X\cap Z$).
The next lemma shows that $J$ can be compressed into a graph $J_X$ on
$Z\cup X$ plus at most one auxiliary vertex such that $J_X$ preserves
PMC membership for all candidates $\Omega$ satisfying
$\Omega\cap\nodes(J)=X$.

\begin{lemma}\label{lem:retained-small}
Let $H,J$ be graphs with $\nodes(H)\cap\nodes(J)=Z$, and fix
$X\subseteq\nodes(J)$. There is a graph $J_X$, constructed from
$J,Z,X$, such that
\[
 Z\cup X\subseteq\nodes(J_X),\qquad
 |\nodes(J_X)|\le|Z\cup X|+1,\qquad
 \nodes(J_X)\cap\nodes(H)=Z.
\]
For every nonempty $\Omega\subseteq\nodes(H)\cup\nodes(J)$
with $\Omega\cap\nodes(J)=X$,
\begin{equation}\label{eq:retained-small}
 \Omega\in\PMC(H\oplus_Z J)
 \quad\Longleftrightarrow\quad
 \Omega\in\PMC(H\oplus_Z J_X).
\end{equation}
Any vertex of $J_X$ outside $Z\cup X$ is a new auxiliary vertex.
\end{lemma}

\begin{proof}
Set $W\eqdef Z\cup X$ and $F\eqdef \Tor(J,W)$. If $X\ne\emptyset$ and
some connected component $C$ of $J\sminus W$ satisfies $N_J(C)=X$, we
define $J_X$ by adding to $F$ a single fresh vertex $z$ (disjoint from
$\nodes(H)\cup\nodes(J)$) adjacent precisely to $X$; otherwise, we set
$J_X=F$. By construction, $Z\cup X\subseteq\nodes(J_X)$,
$|\nodes(J_X)|\le|Z\cup X|+1$, and $\nodes(J_X)\cap\nodes(H)=Z$.

Let $\Omega\subseteq\nodes(H)\cup\nodes(J)$ be any nonempty set with
$\Omega\cap\nodes(J)=X$, so $\Omega\subseteq\nodes(H)\cup W$ and
$\Omega\cap W=X$. Applying Lemma~\ref{lem:torso} with $J$ in the role of
the compressed piece, $W$ as its retained set, and $H$ as the graph
attached along $Z$ yields
\begin{equation}\label{eq:retained-torso}
 \begin{aligned}
 \Omega\in\PMC(H\oplus_Z J)
 \quad\Longleftrightarrow\quad &
 \Omega\in\PMC(H\oplus_Z F)\ \text{and}\\
 &\Omega\ne N_J(C)\quad
 \text{for every component }C\text{ of }J\sminus W.
 \end{aligned}
\end{equation}
Because $N_J(C)\subseteq W$ for every component $C$ of $J\sminus W$,
the equality $\Omega=N_J(C)$ can hold only if $\Omega\subseteq W$, which
forces $\Omega=\Omega\cap W=X$. Thus, among candidates satisfying
$\Omega\cap\nodes(J)=X$, the condition $\Omega\ne N_J(C)$ in~\eqref{eq:retained-torso}
can be violated only when $\Omega=X=N_J(C)$.

If $J_X=F$ (no auxiliary vertex was added), then either $X=\emptyset$
(whereas $\Omega\ne\emptyset$) or no component $C$ of $J\sminus W$ has
$N_J(C)=X$; in either case, $\Omega\ne N_J(C)$ holds automatically for
all components $C$ of $J\sminus W$, and~\eqref{eq:retained-torso}
reduces to~\eqref{eq:retained-small}.

If instead the auxiliary vertex $z$ was added to form $J_X$, then some
component $C$ of $J\sminus W$ has $N_J(C)=X$, so the second condition
in~\eqref{eq:retained-torso} simplifies to $\Omega\ne X$. Because
$X=N_J(C)$, the torso operation turns $X$ into a clique in $F$ (and
hence in $H\oplus_Z F$), and $z\notin\Omega$. Applying
Lemma~\ref{lem:marker} to $H\oplus_Z F$ and the clique $X$ gives
\[
 \Omega\in\PMC(H\oplus_Z J_X)
 \quad\Longleftrightarrow\quad
 \Omega\in\PMC(H\oplus_Z F)\text{ and }\Omega\ne X,
\]
which matches~\eqref{eq:retained-torso} and completes the proof.
\end{proof}

Notice that if we impose the inclusion constraint $I_{J_X}=X$ and
exclusion constraint $O_{J_X}=\nodes(J_X){\setminus}X$ on $J_X$ (which
in particular excludes the auxiliary vertex $z$, if present, as well as
$Z{\setminus}X$), then every constrained PMC $\Omega$ of $H\oplus_Z J_X$
satisfies $\Omega\subseteq\nodes(H)\cup X$ and $\Omega\cap\nodes(J)=X$,
so Lemma~\ref{lem:retained-small} applies in both directions.

\subsection{Retaining one solution for each successful test}
\label{sec:retained-tests}
Lemma~\ref{lem:retained-small} allows us to construct $R\subseteq\nodes(H)$
by testing all possible small attachments $T$ on $Z$ plus at most $k+1$
fresh vertices. Specifically, fix $k+1$ distinct fresh vertex names
$p_1,\ldots,p_{k+1}\notin\nodes(H)$. For each $s\in\{0,\ldots,k+1\}$,
we enumerate every simple graph $T$ on vertex set
$Z\cup\{p_1,\ldots,p_s\}$ (including all possible edge subsets within
$Z$) and every subset $I_T\subseteq\nodes(T)$ with $|I_T|\le k$, setting
$O_T=\nodes(T){\setminus}I_T$. Let $\mathcal{T}_{Z,k}$ denote the family
of all such test triples $(T,I_T,O_T)$. A test triple
$(T,I_T,O_T)\in\mathcal{T}_{Z,k}$ is called \emph{successful} if
\[
 \PMC_k(H\oplus_Z T,I_H\cup I_T,O_H\cup O_T)\ne\emptyset.
\]
Setting $b=|Z|$ and $m=b+k+1$, the number of test triples in
$\mathcal{T}_{Z,k}$ is bounded by
\begin{equation}\label{eq:retained-test-count}
 |\mathcal{T}_{Z,k}|
 =\sum_{s=0}^{k+1}2^{\binom{b+s}{2}}
     \sum_{i=0}^{\min\{k,b+s\}}\binom{b+s}{i}
 \le (k+2)2^{\binom m2+m}=\tau_{b,k},
\end{equation}
where $2^{\binom{b+s}{2}}$ counts the possible edge sets on
$Z\cup\{p_1,\ldots,p_s\}$ and the inner sum counts the choices of
$I_T$.

\begin{lemma}\label{lem:retained-witnesses}
	Let $H$ be a graph, let $k\ge1$ be an integer, and let
	$Z,I_H,O_H\subseteq\nodes(H)$.
	Fix distinct vertices $p_1,\ldots,p_{k+1}$ outside $\nodes(H)$.
	Define $\mathcal{T}_{Z,k}$ to be the family of all triples
	$(T,I_T,O_T)$ where:
	\begin{enumerate}
		\item $T$ is a simple graph on $Z\cup\{p_1,\ldots,p_s\}$
		for some $s\in\{0,\ldots,k+1\}$;
		\item $I_T\subseteq\nodes(T)$ and $|I_T|\le k$;
		\item $O_T=\nodes(T){\setminus} I_T$.
	\end{enumerate}
	A triple $(T,I_T,O_T)$ is \emph{successful} if 
	$\PMC_k(H\oplus_Z T,I_H\cup I_T,O_H\cup O_T)\neq \emptyset$.
	For each successful triple, choose one such PMC $\Omega_{T,I_T}$,
	and put
	\[
	R=Z\cup I_H\cup
	\bigcup_{\substack{(T,I_T,O_T)\in\mathcal{T}_{Z,k}\\
			\text{successful}}}
	(\Omega_{T,I_T}\cap\nodes(H)).
	\]
	For every graph $J$ with $\nodes(H)\cap\nodes(J)=Z$
	and every $I_J,O_J\subseteq\nodes(J)$, if $\PMC_k(H\oplus_Z J,I_H\cup I_J, O_H\cup O_J)\neq \emptyset$, then this family contains
	a PMC $\Omega$ satisfying $\Omega\cap\nodes(H)\subseteq R$.
\end{lemma}

\begin{proof}
Let $J$ be any graph with $\nodes(H)\cap\nodes(J)=Z$, let
$I_J,O_J\subseteq\nodes(J)$, and suppose there exists
$\Omega\in\PMC_k(H\oplus_Z J,I_H\cup I_J,O_H\cup O_J)$. Set
$X=\Omega\cap\nodes(J)$, so $|X|\le|\Omega|\le k$, $I_J\subseteq X$, and
$X\cap O_J=\emptyset$. Let $J_X$ be the graph obtained from $(J,Z,X)$
by Lemma~\ref{lem:retained-small}. Since
$|\nodes(J_X){\setminus}Z|\le|X{\setminus}Z|+1\le k+1$, renaming the
$s=|\nodes(J_X){\setminus}Z|$ vertices of $\nodes(J_X){\setminus}Z$ to
$p_1,\ldots,p_s$ while keeping all vertices of $\nodes(H)$ (and thus of
$Z$) fixed maps $(J_X,X,\nodes(J_X){\setminus}X)$ to a test triple
$(T,I_T,O_T)\in\mathcal{T}_{Z,k}$.

By Lemma~\ref{lem:retained-small}, since $\Omega\in\PMC(H\oplus_Z J)$
and $\Omega\cap\nodes(J)=X$, we have $\Omega\in\PMC(H\oplus_Z J_X)$.
Because $\Omega$ has size at most $k$, satisfies $(I_H,O_H)$ on $H$, and
intersects $\nodes(J_X)$ in $X$, its image under the vertex renaming on
$\nodes(J_X){\setminus}Z$ belongs to
$\PMC_k(H\oplus_Z T,I_H\cup I_T,O_H\cup O_T)$. Thus $(T,I_T,O_T)$ is
successful.

Now let $\Omega_{T,I_T}\in\PMC_k(H\oplus_Z T,I_H\cup I_T,O_H\cup O_T)$
be the witness chosen for $(T,I_T,O_T)$, and let $\Omega'$ be the set
obtained from $\Omega_{T,I_T}$ by applying the inverse renaming from
$\{p_1,\ldots,p_s\}$ back to $\nodes(J_X){\setminus}Z$. Then
$\Omega'\in\PMC_k(H\oplus_Z J_X,I_H\cup X,O_H\cup(\nodes(J_X){\setminus}X))$.
In particular, $\Omega'\cap\nodes(J_X)=X\subseteq\nodes(J)$, so
$\Omega'\subseteq\nodes(H)\cup\nodes(J)$ and $\Omega'\cap\nodes(J)=X$.
By the reverse direction of Lemma~\ref{lem:retained-small},
$\Omega'\in\PMC(H\oplus_Z J)$. Since $|\Omega'|=|\Omega_{T,I_T}|\le k$,
$I_J\subseteq X=\Omega'\cap\nodes(J)$, and
$O_J\cap\Omega'=O_J\cap X=\emptyset$, we have
\[
 \Omega'\in\PMC_k(H\oplus_Z J,I_H\cup I_J,O_H\cup O_J).
\]
Because the bijection between $\nodes(J_X)$ and $\nodes(T)$ is the
identity on $\nodes(H)$, we conclude that
$\Omega'\cap\nodes(H)=\Omega_{T,I_T}\cap\nodes(H)\subseteq R$.
\end{proof}

Notice that if every test in $\mathcal{T}_{Z,k}$ fails,
Lemma~\ref{lem:retained-witnesses} implies that
$\PMC_k(H\oplus_Z J,I_H\cup I_J,O_H\cup O_J)=\emptyset$ for every
outside attachment $J$, so returning $R=Z\cup I_H$ vacuously satisfies
the guarantee.

\subsection{Computing the retained set using \texorpdfstring{$\Local$}{Local}}
\label{sec:retained-algorithm}

Algorithm~\ref{alg:retained} ($\texttt{ChooseR}$) implements the
construction of Lemma~\ref{lem:retained-witnesses} by invoking $\Local$
(Algorithm~\ref{alg:local}) on each test instance $G_T=H\oplus_Z T$ with
unbreakability threshold $q+k+1$.

\begin{algorithm}[htbp]
\caption{$\texttt{ChooseR}(H,Z,k,I_H,O_H,q)$}\label{alg:retained}
\begin{algorithmic}[1]
\Require $H$ is $(q,k)$-unbreakable; $q,k\ge1$;
 $Z,I_H,O_H\subseteq\nodes(H)$; $|I_H|\le k$
\Ensure A retained set $R$ satisfying~\eqref{eq:retained-guarantee}
\State $R\gets Z\cup I_H$
\For{each $(T,I_T,O_T)\in\mathcal{T}_{Z,k}$, generated one at a time}
 \State $G_T\gets H\oplus_Z T$
 \State $I_T^+\gets I_H\cup I_T$; $O_T^+\gets O_H\cup O_T$
 \State $\Omega_{T,I_T}\gets\Local(G_T,k,I_T^+,O_T^+,q+k+1)$
 \If{$\Omega_{T,I_T}\ne\texttt{No}$}
  \State $R\gets R\cup(\Omega_{T,I_T}\cap\nodes(H))$
 \EndIf
\EndFor
\State \Return $R$
\end{algorithmic}
\end{algorithm}

\begin{reptheorem}{\ref{thm:retained}}
	\thmRetained{\eat}
\end{reptheorem}
\begin{proof}
We first verify that $G_T=H\oplus_Z T$ is $(q+k+1,k)$-unbreakable for
every $(T,I_T,O_T)\in\mathcal{T}_{Z,k}$. Let $(L',Z',U')$ be any
separation of $G_T$ of order $|Z'|\le k$. Restricting $(L',Z',U')$ to
$\nodes(H)$ gives a separation $(L'\cap\nodes(H),Z'\cap\nodes(H),U'\cap\nodes(H))$
of $H$ of order at most $k$. Because $H$ is $(q,k)$-unbreakable, at
least one side—say $L'\cap\nodes(H)$—has size at most $q$. Since
$|\nodes(G_T){\setminus}\nodes(H)|\le k+1$, we have
$|L'|\le|L'\cap\nodes(H)|+k+1\le q+k+1$, so $G_T$ is
$(q+k+1,k)$-unbreakable.

By Theorem~\ref{thm:main}, each call to
$\Local(G_T,k,I_T^+,O_T^+,q+k+1)$ in line~5 of
Algorithm~\ref{alg:retained} correctly returns a witness
$\Omega_{T,I_T}\in\PMC_k(G_T,I_T^+,O_T^+)$ whenever the test is
successful and returns $\texttt{No}$ otherwise. Lemma~\ref{lem:retained-witnesses}
therefore ensures that the returned set $R$ satisfies~\eqref{eq:retained-guarantee}.

To bound $|R|$, observe that if at least one test succeeds, then
$I_H\subseteq\Omega_{T,I_T}\cap\nodes(H)$ for every successful test, so
$R$ is the union of $Z$ (of size $b$) with at most $\tau_{b,k}$ sets
$\Omega_{T,I_T}\cap\nodes(H)$ of size at most $k$ each, so $|R|\le b+k\tau_{b,k}$. If every test fails, then
$R=Z\cup I_H$ has size at most $b+k\le b+k\tau_{b,k}$.

Finally, we analyze the running time. Each test graph $T$ has at most
$m=b+k+1$ vertices, and there are at most $\tau_{b,k}=2^{O((b+k)^2)}$
triples in $\mathcal{T}_{Z,k}$. For each triple, constructing $G_T$ on
at most $n+k+1$ vertices and updating $R$ takes $O((n+k+1)^2)$ time, and
by Theorem~\ref{thm:main}, the call to $\Local$ runs in time
\[
 2^{O(k\log(2+(q+k+1)/k))}(n+k+1)^3
 =2^{O(k\log(2+q/k))}(n+k+1)^3,
\]
since $q,k\ge1$ implies $\log(2+(q+k+1)/k)=\Theta(\log(2+q/k))$.
Multiplying by $|\mathcal{T}_{Z,k}|\le\tau_{b,k}$ and using
$n+k+1\le(k+2)n$ for $n\ge1$ yields the bound in~\eqref{eq:retained-time}.
\end{proof}

Combining $\texttt{ChooseR}$ (Algorithm~\ref{alg:retained}) with
$\Compress$ (Algorithm~\ref{alg:compress}) completes our replacement
procedure for any $(q,k)$-unbreakable piece $H$: we first compute
$R\gets\texttt{ChooseR}(H,Z,k,I_H,O_H,q)$ and then return
$(K,I_K,O_K)\gets\Compress(H,Z,R,I_H,O_H)$. By Theorem~\ref{thm:retained}
and Corollary~\ref{cor:compression-constraints}, $K$ has at most
$2(b+k\tau_{b,k})$ vertices, and for every outside graph $J$ attached
along $Z$ and constraints $I_J,O_J\subseteq\nodes(J)$, the instance on
$H\oplus_Z J$ admits a constrained PMC of size at most $k$ if and only
if the instance on $K\oplus_Z J$ does.

\section{Solving \texttt{Extend} on an arbitrary graph}\label{sec:general-algorithm}
When the input graph $G$ is $(q,k)$-unbreakable, Algorithm $\Local$
(Algorithm~\ref{alg:local}) solves \texttt{Extend} directly by
Theorem~\ref{thm:main}. When $G$ is $(q,k)$-breakable, it admits a
separation $(L,Z,U)$ with $|Z|\le k$, $|L|>q$, and $|U|>q$. If one of
the two pieces—say $G[L\cup Z]$—is $(q,k)$-unbreakable, then running
$\texttt{ChooseR}$ (Algorithm~\ref{alg:retained}) followed by
$\Compress$ (Algorithm~\ref{alg:compress}) produces a bounded replacement
graph $K$ such that
\[
\PMC_k(G,I,O)\neq\emptyset
\quad\Longleftrightarrow\quad
\PMC_k\bigl(K\oplus_Z G[U\cup Z],\,I,\,
(O\cap(U\cup Z))\cup O_K\bigr)\neq\emptyset,
\]
which reduces the number of vertices whenever
$|\nodes(K)|<|\nodes(G[L\cup Z])|$, even if the complementary piece
$G[U\cup Z]$ is breakable. In a general graph, however, both pieces of a
separation may be $(q,k)$-breakable.

To handle arbitrary graphs, we design a recursive procedure
$\texttt{Reduce}$ (Algorithm~\ref{alg:general-reduce}) that takes a
piece $H$ (which may already contain excluded auxiliary vertices from
earlier reductions), its inclusion and exclusion constraints $I_H,O_H$,
and its boundary $B\subseteq\nodes(H)$ of size $|B|\le 2k$, and returns
a compressed replacement $(K,I_K,O_K)$ of size $|\nodes(K)|\le r$ such
that $B\subseteq\nodes(K)$ and, for every outside graph $J$ with
$\nodes(H)\cap\nodes(J)=\nodes(K)\cap\nodes(J)=B$ and every choice of
constraints $I_J,O_J\subseteq\nodes(J)$,
\[
\begin{aligned}
	\PMC_k(H\oplus_B J,\,I_H\cup I_J,\,O_H\cup O_J)\neq\emptyset\
	&&\Longleftrightarrow&&
	\PMC_k(K\oplus_B J,\,I_K\cup I_J,\,O_K\cup O_J)\neq\emptyset.
\end{aligned}
\]

Specifically, let $r=r(k)=2^{O(k^2)}$ be the upper bound on the size of
the replacement returned by $\Compress$ when $|B|\le 2k$. If
$|\nodes(H)|\le r$, then $\texttt{Reduce}$ simply returns $(H,I_H,O_H)$
unchanged (line~\ref{line:HisSmall}). Otherwise, $\texttt{Reduce}$
invokes the deterministic cut-finding procedure of Golovach et
al.~\cite[Lemma~2, p.~18:3]{golovach2018} (built on Chitnis et
al.~\cite[Section~3.5]{chitnis2016}), which we denote by
$\texttt{FindCut}$. Given $(H,r,k)$, $\texttt{FindCut}$ either certifies
that $H$ is $(q,k)$-unbreakable for $q=(2r+1)r2^k$, or returns a
separation $(L,Z,U)$ of $H$ with $|Z|\le k$ and $|L|,|U|>r$. If
$\texttt{FindCut}$ certifies that $H$ is $(q,k)$-unbreakable,
$\texttt{Reduce}$ computes $R\gets\texttt{ChooseR}(H,B,k,I_H,O_H,q)$ and
returns $\Compress(H,B,R,I_H,O_H)$
(lines~\ref{line:ifHUnbreakable}--\ref{line:ifHUnbreakableEnd}).

If instead $\texttt{FindCut}$ returns a separation $(L,Z,U)$ with
$|Z|\le k$ and $|L|,|U|>r$, we orient the separation so that
$|B\cap L|\le|B\cap U|$, which guarantees that the new boundary
$B'=(B\cap L)\cup Z$ of $H_L=H[L\cup Z]$ satisfies $|B'|\le 2k$.
In line~\ref{line:firstrRecursiveCall}, $\texttt{Reduce}$ recursively
compresses $H_L$ with boundary $B'$ to obtain a graph $K_L$ on at most
$r$ vertices. Substituting $K_L$ for $H_L$ produces the intermediate
graph $H'=K_L\oplus_Z H[U\cup Z]$ (line~\ref{line:fotmH'}), which has
strictly fewer vertices than $H$ because $|L|>r$. Finally,
line~\ref{line:secondRecursiveCall} makes a second recursive call to
$\texttt{Reduce}$ on $H'$ with the original boundary $B$ and returns the
resulting replacement $(K,I_K,O_K)$.

Applying $\texttt{Reduce}$ at the top level to $(G,k,I,O,\emptyset)$
reduces the entire input graph $G$ in $O(n)$ recursive calls to an
equivalent instance $(K,I_K,O_K)$ on at most $r$ vertices. This overall reduction builds on the recursive-understanding scheme of Chitnis et al.~\cite[Section~2.1]{chitnis2016} and Lokshtanov et al.~\cite{lokshtanov2018}, and searching $K$ directly establishes our main theorem.

\begin{theorem}[Constrained PMC search on arbitrary graphs]
\label{thm:general-extend}
Let $G$ be an $n$-vertex graph, let $k\ge0$ be an integer, and let
$I,O\subseteq\nodes(G)$. There is a deterministic algorithm
$\texttt{Extend}$ that returns a set
$\Omega\in\PMC_k(G,I,O)$, or correctly reports that no such set exists.
For $n\ge2$, its running time is
\[
 2^{O(k^3)}n^4\log n.
\]
Inputs with at most one vertex are handled in constant time.
No unbreakability assumption is imposed on $G$.
\end{theorem}

\subsection{Replacement size and separation threshold}
\label{sec:general-parameters}

Throughout the recursion, each call to $\texttt{Reduce}$ receives a
piece $H$, constraint sets $I_H,O_H\subseteq\nodes(H)$ with $|I_H|\le k$,
and a boundary set $B\subseteq\nodes(H)$ satisfying $|B|\le 2k$.
Substituting $b=2k$ into the bound of Theorem~\ref{thm:retained}, we
define the replacement size bound $r$ by
\begin{equation}\label{eq:general-r}
 \tau_{2k,k}=(k+2)2^{\binom{3k+1}{2}+3k+1},
 \qquad r=2(2k+k\tau_{2k,k})=2^{O(k^2)}.
\end{equation}
Whenever $H$ is certified to be $(q,k)$-unbreakable, invoking
$R\gets\texttt{ChooseR}(H,B,k,I_H,O_H,q)$ followed by
$\Compress(H,B,R,I_H,O_H)$ produces a graph $K$ with
$|\nodes(K)|\le 2|R|\le r$ and constraint sets
\begin{equation}\label{eq:general-base-constraints}
 I_K=I_H,\qquad
 O_K=(O_H\cap R)\cup(\nodes(K){\setminus} R).
\end{equation}
By Theorem~\ref{thm:retained} and Corollary~\ref{cor:compression-constraints},
for every outside graph $J$ attached along $B$, the replacement
$(K,I_K,O_K)$ preserves the existence of a constrained PMC of size at
most $k$; in addition, because all auxiliary vertices
$\nodes(K){\setminus}R$ belong to $O_K$, every constrained PMC of
$K\oplus_B J$ is already a constrained PMC of $H\oplus_B J$ as the exact
same vertex set.

Notice that the replacement bound $r$ in~\eqref{eq:general-r} depends
only on $k$ and is independent of the unbreakability threshold $q$. We
can therefore fix $r$ first and then set $q$ according to the
unbreakability guarantee of the cut-finding algorithm of Chitnis et
al.~\cite[Section~3.5]{chitnis2016}, stated in the following form by
Golovach et al.~\cite[Lemma~2, p.~18:3]{golovach2018}.

\begin{citedLemma}{\cite[Lemma~2, p.~18:3]{golovach2018}}
\label{lem:general-cut}
There is a deterministic algorithm that, given an $n$-vertex
graph $H$, where $n\ge2$, and two integers $r,k\ge1$, either returns a separation $(L,Z,U)$
with
\[
 |Z|\le k,\qquad |L|>r,\qquad |U|>r,
\]
or correctly reports that $H$ is $((2r+1)r2^k,k)$-unbreakable.
Its running time is
\[
 2^{O(\min\{r,k\}\log(r+k))}n^3\log n.
\]
\end{citedLemma}

We refer to the algorithm of Lemma~\ref{lem:general-cut} as
$\texttt{FindCut}(H,r,k)$, and fix the global unbreakability threshold
\begin{equation}\label{eq:general-q}
 q\eqdef (2r+1)r2^k=2^{O(k^2)}.
\end{equation}

\subsection{Recursive replacement}\label{sec:general-reduce}

Algorithm~\ref{alg:general-reduce} ($\texttt{Reduce}$) gives the full
recursive replacement procedure. When $|\nodes(H)|>r$ and
$\texttt{FindCut}(H,r,k)$ returns a separation $(L,Z,U)$ with $|Z|\le k$
and $|L|,|U|>r$, we swap $L$ and $U$ if necessary so that
$|B\cap L|\le|B\cap U|$. Because $L,Z,U$ partition $\nodes(H)$ and
$|B|\le 2k$, we have
\[
 |B\cap L|\le \tfrac12|B{\setminus} Z|\le  \tfrac12|B| \le k,
 \qquad |B'|=|(B\cap L)\cup Z|=|B\cap L|+|Z|\le2k.
\]
In line~\ref{line:firstrRecursiveCall}, the first recursive call
compresses $H_L=H[L\cup Z]$ with boundary $B'=(B\cap L)\cup Z$, which
retains both $Z$ (the interface to $H[U\cup Z]$) and $B\cap L$ (the
portion of the external boundary $B$ lying in $L$). After substituting
the returned graph $K_L$ for $H_L$ to form
$H'=K_L\oplus_Z H[U\cup Z]$ in line~\ref{line:fotmH'}, the second
recursive call in line~\ref{line:secondRecursiveCall} compresses $H'$
with the original boundary $B$. As always, every auxiliary vertex created
during any call to $\Compress$ receives a fresh identifier distinct from
all original and previously created vertices.

\begin{algorithm}[htbp]
\caption{$\texttt{Reduce}(H,k,I_H,O_H,B)$}\label{alg:general-reduce}
\begin{algorithmic}[1]
\Require $k\ge1$, $I_H,O_H,B\subseteq\nodes(H)$,
 $|I_H|\le k$, $|B|\le2k$; $r,q$ as above
\Ensure A graph $K$ of size at most $r$ and constraints $I_K,O_K$
 satisfying Lemma~\ref{lem:general-reduce}
\State \textbf{if} $|\nodes(H)|\le r$ \textbf{then return} $(H,I_H,O_H)$ \label{line:HisSmall}
\State run $\texttt{FindCut}(H,r,k)$
\If{it certifies that $H$ is $(q,k)$-unbreakable}\label{line:ifHUnbreakable}
 \State $R\gets\texttt{ChooseR}(H,B,k,I_H,O_H,q)$
 \State \Return $\Compress(H,B,R,I_H,O_H)$
\EndIf \label{line:ifHUnbreakableEnd}
\State let $(L,Z,U)$ be the returned separation
\State exchange $L,U$ if necessary so that $|B\cap L|\le|B\cap U|$
\State $H_L\gets H[L\cup Z]$; $B'\gets(B\cap L)\cup Z$
\State $I_L\gets I_H\cap(L\cup Z)$; $O_L\gets O_H\cap(L\cup Z)$
\State $(K_L,I_{K_L},O_{K_L})\gets
 \texttt{Reduce}(H_L,k,I_L,O_L,B')$ \label{line:firstrRecursiveCall}
\State $H'\gets K_L\oplus_Z H[U\cup Z]$ \label{line:fotmH'}
\State $I'\gets I_{K_L}\cup(I_H\cap(U\cup Z))$
\State $O'\gets O_{K_L}\cup(O_H\cap(U\cup Z))$
\State \Return $\texttt{Reduce}(H',k,I',O',B)$ \label{line:secondRecursiveCall}
\end{algorithmic}
\end{algorithm}

\subsection{Correctness of the recursive replacement}
\label{sec:general-correctness}

We now prove by strong induction on $|\nodes(H)|$ that $\texttt{Reduce}$
satisfies its specification at every level of the recursion.

\begin{lemma}[Correctness of $\texttt{Reduce}$]\label{lem:general-reduce}
Let $H$ be a graph, let $k\ge1$ be an integer, and let
$I_H,O_H,B\subseteq\nodes(H)$ with $|I_H|\le k$ and $|B|\le2k$.
With $r,q$ defined by~\eqref{eq:general-r} and~\eqref{eq:general-q},
$\texttt{Reduce}(H,k,I_H,O_H,B)$ terminates and returns
$(K,I_K,O_K)$ such that
\[
 B\cup I_H\subseteq\nodes(K),\qquad |\nodes(K)|\le r,
 \qquad I_K=I_H,
\]
and
\begin{equation}\label{eq:general-forbidden}
 O_K=(O_H\cap\nodes(K))\cup(\nodes(K){\setminus}\nodes(H)).
\end{equation}
For every graph $J$ sharing exactly $B$ with both $H$ and $K$,
and every $I_J,O_J\subseteq\nodes(J)$:
\begin{enumerate}
 \item The graph $H\oplus_B J$ has a PMC of size at most $k$
 containing $I_H\cup I_J$ and avoiding $O_H\cup O_J$ if and only if
 $K\oplus_B J$ has a PMC of size at most $k$ containing
 $I_K\cup I_J$ and avoiding $O_K\cup O_J$; that is, $\PMC_k(H\oplus_B J,I_H\cup I_J,O_H\cup O_J)\neq \emptyset \Leftrightarrow \PMC_k(K\oplus_B J,I_K\cup I_J,O_K\cup O_J)\neq \emptyset$.
 \item Every
 $\Omega\in\PMC_k(K\oplus_B J,I_K\cup I_J,O_K\cup O_J)$
 belongs to
 $\PMC_k(H\oplus_B J,I_H\cup I_J,O_H\cup O_J)$
 as the same set of vertices.
\end{enumerate}
\end{lemma}

\begin{proof}
We proceed by strong induction on $n=|\nodes(H)|$. If $n\le r$,
line~\ref{line:HisSmall} returns $(H,I_H,O_H)$ unchanged and all claims
hold trivially. If $n>r$ and $\texttt{FindCut}$ certifies that $H$ is
$(q,k)$-unbreakable, the claims follow directly from
Theorem~\ref{thm:retained}, Corollary~\ref{cor:compression-constraints},
and~\eqref{eq:general-base-constraints} (since
$\nodes(K)\cap\nodes(H)=R$).

Suppose now that $n>r$ and $\texttt{FindCut}$ returns a separation
$(L,Z,U)$ with $|Z|\le k$ and $|L|,|U|>r$. By the choice of $L$ in
line~8, the recursive call in line~\ref{line:firstrRecursiveCall}
satisfies $|B'|\le 2k$ and $|I_L|\le|I_H|\le k$, and since $|U|>r\ge 1$,
we have $|\nodes(H_L)|=|L|+|Z|=n-|U|<n$. By the induction hypothesis,
the first recursive call terminates and returns $(K_L,I_{K_L},O_{K_L})$
with $|\nodes(K_L)|\le r$, $B'\cup I_L\subseteq\nodes(K_L)$,
$I_{K_L}=I_L$, and
$O_{K_L}=(O_L\cap\nodes(K_L))\cup(\nodes(K_L){\setminus}\nodes(H_L))$,
satisfying guarantees~(1) and~(2) with respect to $H_L$ and boundary
$B'$.

Next, consider the graph $H'=K_L\oplus_Z H[U\cup Z]$ formed in
line~\ref{line:fotmH'}. Because $\nodes(K_L)\cap\nodes(H[U\cup Z])=Z$
and $s\eqdef|\nodes(K_L)|\le r<|L|$, the number of vertices in $H'$ is
\[
 |\nodes(H')|=s+|U|\le r+|U|
 <|L|+|Z|+|U|=n.
\]
All vertices of $B$ belong to $\nodes(H')$, since
$B\cap L\subseteq B'\subseteq\nodes(K_L)$ and
$B\cap(U\cup Z)\subseteq\nodes(H[U\cup Z])$. Similarly,
$I'=I_{K_L}\cup(I_H\cap(U\cup Z))=I_H$ has size at most $k$. Thus the
second recursive call $\texttt{Reduce}(H',k,I',O',B)$ in
line~\ref{line:secondRecursiveCall} satisfies all preconditions on a
graph with strictly fewer than $n$ vertices, and by the induction
hypothesis it terminates and returns $(K,I_K,O_K)$ with
$|\nodes(K)|\le r$, $B\cup I'\subseteq\nodes(K)$, $I_K=I'=I_H$, and
guarantees~(1) and~(2) relative to $H'$.

\paragraph{Replacing $H_L$ in the attached graph.}
Let $J$ be any outside graph sharing $B$ with $H$ and $K$, with
constraints $I_J,O_J\subseteq\nodes(J)$ (with $\nodes(J){\setminus}B$
disjoint from intermediate auxiliary vertices as noted in
Section~\ref{sec:preliminaries}). Observe that attaching $J$ to $H$
along $B$ is identical to attaching the combined outside graph
\[
 C=H[U\cup Z]\oplus_{B\cap(U\cup Z)}J,
\]
with constraint sets
\[
 I_C=(I_H\cap(U\cup Z))\cup I_J,
 \qquad O_C=(O_H\cap(U\cup Z))\cup O_J,
\]
to the piece $H_L=H[L\cup Z]$ along the boundary $B'=(B\cap L)\cup Z$:
indeed, $\nodes(H_L)\cap\nodes(C)=\nodes(K_L)\cap\nodes(C)=B'$, and
\[
 H_L\oplus_{B'}C=H\oplus_B J,
 \qquad K_L\oplus_{B'}C=H'\oplus_B J.
\]
Applying the induction hypothesis for the first recursive call to the
attached graph $C$ and constraints $(I_C,O_C)$ yields
\begin{equation}\label{eq:general-first-replacement}
 \begin{aligned}
 \PMC_k(H\oplus_B J,I_H\cup I_J,O_H\cup O_J)\ne\emptyset
 &&\Longleftrightarrow&&
 \PMC_k(H'\oplus_B J,I'\cup I_J,O'\cup O_J)\ne\emptyset,
 \end{aligned}
\end{equation}
and ensures that every member of the right-hand family belongs to the
left-hand family as the same vertex set.

\paragraph{Replacing $H'$ by the final output $K$.}
Combining~\eqref{eq:general-first-replacement} with guarantee~(1) of the
induction hypothesis for the second recursive call on $H'$ immediately
gives
\[
\begin{aligned}
	&\PMC_k(H\oplus_B J,I_H\cup I_J,O_H\cup O_J)\neq\emptyset &&\Longleftrightarrow&&
	\PMC_k(K\oplus_B J,I_K\cup I_J,O_K\cup O_J)\neq\emptyset,
\end{aligned}
\]
which proves guarantee~(1) for $H$. Likewise, for any
$\Omega\in\PMC_k(K\oplus_B J,I_K\cup I_J,O_K\cup O_J)$, guarantee~(2)
for the second call implies
$\Omega\in\PMC_k(H'\oplus_B J,I'\cup I_J,O'\cup O_J)$, and guarantee~(2)
for the first call (with attachment $C$) then implies
$\Omega\in\PMC_k(H\oplus_B J,I_H\cup I_J,O_H\cup O_J)$, which establishes guarantee~(2) for $H$. Finally, because fresh auxiliary names prevent any
deleted vertex of $H_L$ from reappearing in $K$ and both calls preserve
the exclusion status of retained vertices while excluding all newly
added auxiliary vertices, we have
$O_K=(O_H\cap\nodes(K))\cup(\nodes(K){\setminus}\nodes(H))$. This completes the induction.
\end{proof}

\subsection{Returning a PMC of the original graph}\label{sec:general-search}

Algorithm~\ref{alg:general-extend} ($\texttt{Extend}$) solves
constrained PMC search on an arbitrary input graph $G$ by invoking
$\texttt{Reduce}(G,k,I,O,\emptyset)$ in line~\ref{line:callReduce} to
obtain a compressed graph $(K,I_K,O_K)$ on at most $r$ vertices, and
then testing all nonempty subsets $\Omega\subseteq\nodes(K)$ of size at
most $k$ using $\CheckPMC(K,\Omega,k,I_K,O_K)$.

\begin{algorithm}[htbp]
\caption{$\texttt{Extend}(G,k,I,O)$ on an arbitrary graph}
\label{alg:general-extend}
\begin{algorithmic}[1]
\Require $k\ge0$ and $I,O\subseteq\nodes(G)$
\Ensure A set in $\PMC_k(G,I,O)$, or $\texttt{No}$ if this family is empty
\If{$\nodes(G)=\emptyset$ or $k=0$ or $I\cap O\ne\emptyset$ or $|I|>k$}
 \State \Return $\texttt{No}$
\EndIf
\State $k\gets\min\{k,|\nodes(G)|\}$; fix $r,q$ by
 \eqref{eq:general-r} and~\eqref{eq:general-q}
\State $(K,I_K,O_K)\gets\texttt{Reduce}(G,k,I,O,\emptyset)$ \label{line:callReduce}
\For{each nonempty $\Omega\subseteq\nodes(K)$ with $|\Omega|\le k$}
 \If{$\CheckPMC(K,\Omega,k,I_K,O_K)=\texttt{Yes}$}
  \State \Return $\Omega$
 \EndIf
\EndFor
\State \Return $\texttt{No}$
\end{algorithmic}
\end{algorithm}

Applying Lemma~\ref{lem:general-reduce} to $\texttt{Reduce}(G,k,I,O,\emptyset)$
with $J$ as the empty graph shows that $\PMC_k(G,I,O)\ne\emptyset$ if
and only if $\PMC_k(K,I_K,O_K)\ne\emptyset$, and that every
$\Omega\in\PMC_k(K,I_K,O_K)$ belongs to $\PMC_k(G,I,O)$ as the exact
same vertex subset of $\nodes(G)$. Consequently,
Algorithm~\ref{alg:general-extend} is correct.

\subsection{Running time of \texorpdfstring{$\texttt{Extend}$}{Extend}}\label{sec:general-time}

\paragraph{Number of recursive calls.}
Consider any call to $\texttt{Reduce}$ on an $n$-vertex graph $H$ (with
$n>r$) in which $\texttt{FindCut}$ returns a separation $(L,Z,U)$. Let
\[
 a=|L|,\quad u=|U|,\quad z=|Z|,\quad
 s=|\nodes(K_L)|\le r.
\]
Then $n=a+u+z$, while the two recursive calls in
lines~\ref{line:firstrRecursiveCall} and~\ref{line:secondRecursiveCall}
operate on graphs $H_L$ and $H'$ of sizes $n_1=a+z$ and $n_2=u+s$,
respectively. Since $a,u>r$ and $s\le r$, we have $n_1,n_2>r$ and
\begin{equation}\label{eq:general-call-count}
 (n_1-r)+(n_2-r)=n+s-2r\le n-r.
\end{equation}
By induction on $n>r$, any call to $\texttt{Reduce}$ on an $n$-vertex
graph with $n>r$ makes at most $2(n-r)-1$ calls to $\texttt{Reduce}$ in
total (counting the call itself): for a leaf call this is $1\le 2(n-r)-1$,
and for an internal call, inequality~\eqref{eq:general-call-count} and
the induction hypothesis give at most
\[
 1+\bigl(2(n_1-r)-1\bigr)+\bigl(2(n_2-r)-1\bigr)
 \le2(n-r)-1
\]
calls. Thus, for an $N$-vertex input graph $G$,
$\texttt{Reduce}(G,k,I,O,\emptyset)$ makes $O(N)$ recursive calls, each
on a graph of at most $N$ vertices.

\paragraph{Running time per call.}
By~\eqref{eq:general-r} and~\eqref{eq:general-q}, we have
\[
 \log r=O(k^2),\qquad \log q=O(k^2).
\]
In each call to $\texttt{Reduce}$ on an $n$-vertex graph ($n\ge2$),
running $\texttt{FindCut}(H,r,k)$ takes $2^{O(k^3)}n^3\log n$ time by
Lemma~\ref{lem:general-cut}. When $\texttt{FindCut}$ certifies
$(q,k)$-unbreakability, Theorem~\ref{thm:retained} (with $|B|\le 2k$)
bounds the running time of $\texttt{ChooseR}$ by
\[
2^{O(k^2+k\log(2+q/k))}n^3 = 2^{O(k^3)}n^3,
\]
since $\log(2+q/k)\le\log(2+q)=O(k^2)$, and the subsequent call to
$\Compress$ takes $O(n^3)$ time by Theorem~\ref{thm:compression}. In a
branching call, forming $H_L$, $H'$, and the updated constraint sets
takes $O(n^2)$ time. Multiplying the per-call bound of
$2^{O(k^3)}N^3\log N$ by the $O(N)$ recursive calls bounds the total
time of $\texttt{Reduce}(G,k,I,O,\emptyset)$ by
\[
 2^{O(k^3)}N^4\log N
\]
for $N\ge2$.

\paragraph{Dependence of the final search on $k$.}
Finally, the reduced graph $K$ has at most $r$ vertices, so lines~6--10
of Algorithm~\ref{alg:general-extend} test at most
\[
 \sum_{i=1}^{\min\{k,r\}}\binom ri\le k r^k=2^{O(k^3)}
\]
subsets $\Omega\subseteq\nodes(K)$, each in time
$O(r^2+rk^2)\subseteq O(r^3)=2^{O(k^2)}$ via $\CheckPMC$
(Section~\ref{sec:recognition}). The final search therefore takes $2^{O(k^3)}$ time. Combined with the correctness of $\texttt{Extend}$, this proves Theorem~\ref{thm:general-extend}.

\end{document}